\newif\iffinal
\finaltrue 

\documentclass[reqno,twoside,11pt]{amsart}
\iffinal\else\usepackage[notref,notcite]{showkeys}\fi

\usepackage{amsmath}
\usepackage{amsfonts}
\usepackage{amssymb}
\usepackage{color}
\usepackage[colorlinks=true, linkcolor=black, citecolor=black,pdfborder={0 0 0}]{hyperref}
\usepackage{multido}
\usepackage{url}
\usepackage{mathrsfs}
\usepackage{enumerate}
\usepackage{verbatim}
\usepackage{setspace}

\usepackage{cleveref}

\crefname{theorem}{Theorem}{Theorems}
\crefname{lemma}{Lemma}{Lemmas}
\crefname{proposition}{Proposition}{Propositions}
\crefname{section}{\S}{\S\S}
\crefname{equation}{}{}

\IfFileExists{myowntimes.sty}{\usepackage{myowntimes}\usepackage{mathrsfs}}
        {\usepackage{mathpazo}\usepackage{mathrsfs}}

\DeclareFontFamily{OT1}{eusb}{} \DeclareFontShape{OT1}{eusb}{m}{n}
{<5> <6> <7> <8> <9> <10> <11> <12> <14.4> eusb10}{}
\DeclareMathAlphabet{\eusb}{OT1}{eusb}{m}{n}

\DeclareFontFamily{OT1}{eusm}{} \DeclareFontShape{OT1}{eusm}{m}{n}
{<5> <6> <7> <8> <9> <10> <11> <12> <14.4> eusm10}{}
\DeclareMathAlphabet{\eusm}{OT1}{eusm}{m}{n}

\DeclareFontFamily{OT1}{eufm}{} \DeclareFontShape{OT1}{eufm}{m}{n}
{<5> <6> <7> <8> <9> <10> <11> <12> <14.4> eufm10}{}
\DeclareMathAlphabet{\mathfrak}{OT1}{eufm}{m}{n}

\DeclareFontFamily{OT1}{fraktura}{}
\DeclareFontShape{OT1}{fraktura}{m}{n} {<5> <6> <7> <8> <9> <10> <11>
  <12> <13> <14.4> [1.1] eufm10}{}
\DeclareMathAlphabet{\fraktura}{OT1}{fraktura}{m}{n}

\DeclareFontFamily{OT1}{cmfi}{} \DeclareFontShape{OT1}{cmfi}{m}{n}
{<5> <6> <7> <8> <9> <10> <11> <12> <13> <14.4> [0.9] cmfi10}{}
\DeclareMathAlphabet{\cmfi}{OT1}{cmfi}{b}{n}

\DeclareFontFamily{OT1}{cmss}{} \DeclareFontShape{OT1}{cmss}{m}{n}
{<5> <6> <7> <8> <9> <10> <11> <12> <13> <14.4> cmss10}{}
\DeclareMathAlphabet{\cmss}{OT1}{cmss}{m}{n}

\newtheoremstyle{thm}{1.5ex}{1.5ex}{\itshape\rmfamily}{}
{\bfseries\rmfamily}{}{2ex}{}

\newtheoremstyle{def}{1.5ex}{1.5ex}{\slshape\rmfamily}{}
{\bfseries\rmfamily}{}{2ex}{}

\newtheoremstyle{rem}{1.3ex}{1.3ex}{\rmfamily}{}
{\itshape}
{} {1.5ex}{}

\theoremstyle{thm}
\newtheorem{theorem}{Theorem}[section]
\newtheorem{lemma}[theorem]{Lemma}
\newtheorem{proposition}[theorem]{Proposition}

\newtheorem*{Main Theorem}{Main Theorem.}
\newtheorem{corollary}[theorem]{Corollary}
\newtheorem*{special theorem}{Lindeberg-Feller Theorem for Martingales}

\newtheorem{definition}[theorem]{Definition}

\theoremstyle{rem}
\newtheorem{remark}{{\itshape Remark}}[]

\numberwithin{equation}{section}

\renewcommand{\section}{\secdef\sct\sect}
\newcommand{\sct}[2][default]{%
\refstepcounter{section}
\addcontentsline{toc}{section}{{\tocsection
    {}{\thesection}{\!\!\!\!#1\dotfill}}{}}
\vspace{0.7cm}
\centerline{\scshape\thesection.\ #1} \nopagebreak \vspace{0.2cm}}
\newcommand{\sect}[1]{%
\vspace{0.4cm} \centerline{\large\scshape\rmfamily #1}
\vspace{0.2cm}
}

\renewcommand{\subsection}{\secdef\subsct\sbsect}
\newcommand{\subsct}[2][default]{\refstepcounter{subsection}
\addcontentsline{toc}{subsection}
{{\tocsection{\!\!}{\hspace{1.2em}\thesubsection}{\!\!\!\!#1\dotfill}}{}}
\nopagebreak
{\flushleft\bf
\thesubsection~\bf #1.~}
\noindent
\nopagebreak}
\newcommand{\sbsect}[1]{
\noindent
\textbf{#1.~}
}

\renewcommand{\subsubsection}{%
\secdef \subsubsect\sbsbsect}
\newcommand{\subsubsect}[2][default]{%
\refstepcounter{subsubsection}
\addcontentsline{toc}{subsubsection}{{\tocsection{\!\!}
{\hspace{3.05em}\thesubsubsection}{\!\!\!\!#1\dotfill}}{}}
\nopagebreak
\vspace{0.15\baselineskip} \nopagebreak {\flushleft\rmfamily
\itshape\thesubsubsection
\ \rmfamily #1\/.}\ }
\newcommand{\sbsbsect}[1]{\vspace{0.1cm}\noindent
\rmfamily \itshape
\arabic{section}.\arabic{subsection}.\arabic{subsubsection} \
\sffamily #1\/.\ }

\iffinal

\else

\fi

\renewcommand{\caption}[1]{%
\vglue0.5cm
\refstepcounter{figure}
\begin{minipage}{0.9\textwidth}\small {\sc Figure~\thefigure. }#1\end{minipage}}

\newcommand{\ra}{{e_1}}

\def\qed{ \hfill $\square$}

 \DeclareMathOperator{\E}{\mathbb{E}}
 
 \DeclareMathOperator{\sgn}{sgn}

 \newcommand{\bP}{\mathbb{P}}

\newcommand{\dist}{\operatorname{dist}}

\newcommand{\diam}{\operatorname{diam}}

\newcommand{\Int}{{\text{\rm Int}}}
\newcommand{\Ext}{{\text{\rm Ext}}}

\newcommand{\textd}{\text{\rm d}\mkern0.5mu}

\newcommand{\Sp}{{\text{\rm sp}}}

\renewcommand{\AA}{\mathcal A}
\newcommand{\BB}{\mathcal B}

\newcommand{\DD}{\mathcal D}
\newcommand{\EE}{\mathcal E}
\newcommand{\FF}{\mathcal F}

\newcommand{\HH}{\mathcal H}
\newcommand{\II}{\mathcal I}

\newcommand{\KK}{\mathcal K}

\newcommand{\MM}{\mathcal M}
\newcommand{\NN}{\mathcal N}

\newcommand{\QQ}{\mathcal Q}
\newcommand{\RR}{\mathcal R}

\newcommand{\D}{\mathbb D}

\newcommand{\N}{\mathbb N}

\newcommand{\Q}{\mathbb Q}
\newcommand{\R}{\mathbb R}

\newcommand{\X}{\mathbb X}

\newcommand{\Z}{\mathbb Z}

\def\myffrac#1#2 in #3{\raise 2.6pt\hbox{$#3 #1$}\mkern-1.5mu\raise 0.8pt\hbox{$#3/$}\mkern-1.1mu\lower 1.5pt\hbox{$#3 #2$}}
\newcommand{\ffrac}[2]{\mathchoice%
        {\myffrac{#1}{#2} in \scriptstyle}
        {\myffrac{#1}{#2} in \scriptstyle}
        {\myffrac{#1}{#2} in \scriptscriptstyle}
        {\myffrac{#1}{#2} in \scriptscriptstyle}
}

\newcommand{\pI}{\partial^{\textrm{i}}}
\newcommand{\pO}{\partial^{\textrm{o}}}
\newcommand{\gld}{g^{\lambda, D}}
\newcommand{\gll}{g^{\lambda}}
\newcommand{\gln}{g^{\lambda, N}}
\newcommand{\Gld}{G^{\lambda, D}}
\newcommand{\Gln}{G^{\lambda, N}}
\newcommand{\Gl}{G^{\lambda}}

\newcommand{\ls}{\lesssim}

\author[N. Crawford]{Nicholas Crawford} \thanks{Supported in part at the
    Technion by a Marilyn and Michael Winer Fellowship and a
    Landau Fellowship}
\begin{document}

\title[Random Field Induced Order]{Random Field Induced
Order in Low Dimension I}
\maketitle

\centerline{\textit{Department of Mathematics, The Technion, Haifa, Israel}}

\vspace{2mm} \begin{quote} \footnotesize \textbf{Abstract:} Consider
the classical $XY$ model in a weak random external field pointing
along the $Y$ axis with strength $\epsilon$.  We prove that the model defined on $\Z^3$ with nearest neighbor coupling exhibits residual magnetic order \textit{in the horizontal direction} for arbitrarily weak random field strengths and, depending on field strength, sufficiently low temperature.  \end{quote}
\vspace{2mm}

\section{Introduction}
In this paper we study an interesting phenomena which leads to ordering at low temperatures for spin systems with continuous symmetries:  random field induced ordering (RFO).  To fix ideas and introduce the model which is studied below, consider a classical $O(2)$ model on $\Z^d$.  For each $x \in \Z^d$ let $\sigma_x \in \mathbb S^1$.   A vector of spins $\sigma=(\sigma_x)_{x \in \Z^d}$ will be called a spin configuration.  Let $(\alpha_x(\omega))_{x \in \Z^d}$ be an auxiliary i.i.d. family of standard normal random variables, with $\omega$ representing an element of an auxiliary sample space $\Omega$ on which the $\alpha_x$'s are defined.  For any fixed  spin configuration $\sigma^0_x$ and any bounded region $\Lambda \subset \Z^d$, we can then define the (random) Hamiltonian via
\begin{equation}
\label{E:Ham0}
-\HH_{\Lambda}^{\omega}(\sigma|\sigma^0)= -\frac 12\sum_{\langle x y\rangle \cap \Lambda \neq \varnothing } [\sigma_x- \sigma_y]^2 + \epsilon \sum_{x \in \Lambda} \alpha_x(\omega) e_2 \cdot \sigma_x
\end{equation}
where we set $\sigma_x=\sigma_x^0$ for $x\in \Lambda^c$ and allow $\sigma_x \in \mathbb S^1$ to be arbitrary within $\Lambda$.
Here $\langle xy \rangle$ indicates that $x, y$ are nearest neighbors with respect to the usual graph structure on $\Z^d$ and $e_1, e_2$ denote the standard orhonormal basis of $\R^2$.
Defining finite volume Gibbs measures by
\[
\mu_{\Lambda}^{\sigma_0}(A) =Z_{\Lambda}^{-1}\int_{A} \prod_{x \in \Lambda} \textd \nu(\sigma_x) \exp\{ - \beta \HH_{\Lambda}^{\omega}(\sigma|\sigma^0)\}
\]
and denoting the corresponding Gibbs state by $\langle \cdot \rangle= \langle \cdot \rangle_{\Lambda}^{\sigma_0}$,
the question is whether and when, in terms of $\beta, \epsilon$ and $\sigma_0$, residual magnetic ordering occurs in the limit as $\Lambda \uparrow \Z^d$.  Below we discuss in more detail previous work on the subject, but prior to the present work it was expected that ordering does occur in dimension $d \geq 3$ while in dimension $d=2$ there was uncertainty about the low temperature behavior \cite{ALL, DF1, DF2, Wehr-et-al-2}.  In general, when ordering does occur, the ordering is expected in the horizontal direction $\pm e_1$ \textit{whenever} the projection $\sigma^0_x \cdot e_1$ is either uniformly positive or uniformly negative on $\Z^d$.  In this paper we demonstrate that ordering occurs in the $\pm e_1$ direction only if $d=3$ and only if the boundary condition $\sigma^0_x \equiv \pm e_1$.   Our framework can be extended to prove the same statement for the two dimensional system and also to treat 
more general boundary conditions $\sigma^0_x \equiv u$ for some $u \in \mathbb S^1$.  Because these statements  increase the technical complexity substantially without shedding further light on the methods, we leave them for future work.  When possible, we will comment on modifications which are needed to treat these extensions.

Despite the fact that we only provide a proof in three dimensions, the result is most interesting if $d=2$.  To explain why, we recall the behavior of related models.
\begin{itemize}
\item
\textbf{The Pure $O(2)$ Model.}
In this case we take the same setup as \cref{E:Ham0} except that we set $\epsilon=0$.  When $d=2$ there is no residual magnetic order in the thermodynamic limit (this is the content of the Mermin-Wagner theorem, see \cite{MW, MS, DS, FP} among many other works).  There is however a Kosterlitz-Thouless phase transition \cite{FS} expressed by a change in the behavior of the decay of the spin-spin correlation function
 $\langle \sigma_x \cdot \sigma_y\rangle$ with $|x-y|$.  If $d\geq 3$, residual magnetic ordering occurs \cite{FSS}.

\item
\textbf{The Random Field Ising Model.}
In this case we constrain the spins to be $\pm e_2$, replace $\textd \nu$ by unbiased counting measure and otherwise retain the setup of \cref{E:Ham0}.
When $d \geq 3$,  residual magnetic ordering occurs for $\pm e_2$ boundary conditions.  This was proved for ground states in a system with weak disorder in \cite{Imbrie} and for low temperature and weak disorder in \cite{BK}.   On the other hand in dimension $d=2$ it was proved in \cite{AW}, see also \cite{IM}, that there is a unique infinite volume Gibbs state at arbitrary strength $\epsilon$ of the disorder.

\item\textbf{The Random Field Gaussian Model.}
In this example we replace the vector valued spins $\sigma_x \in \mathbb S^1$ by a scalar field $\phi_x \in \R$, replacing $e_2$ by $1$, and replacing $\textd \nu$ by Lebesgue measure on $\R$.  The model appears in \cite{vE0}, see also \cite{SVBO} for related work. Because the underlying configuration space is no longer compact, the existence of infinite volume Gibbs states is a (somewhat) nontrivial issue and, to an extent, is the analog of the question of residual ordering for compact spin spaces with a continuous symmetry.

To get an indication as to what should be expected we compute two point correlations in finite volume. Fixing a finite volume $\Lambda$, let $-\Delta_{\Lambda}$ denote the discrete Laplace operator on $\Lambda$ with Dirichlet boundary conditions.  Taking boundary conditions $\phi^0_x \equiv 0$ and averaging over the $\alpha_z$'s gives
\[
\E[\langle\phi_x \phi_y \rangle^{\phi^0 \equiv 0}_{\Lambda}]= -\Delta_{\Lambda}^{-1}(x, y)+ \epsilon^2 \Delta^{-2}(x, y).
\]
If $\epsilon$ is $0$, these correlations are finite (uniformly in $\Lambda$) only if $d\geq 3$.  In this case one can define infinite volume Gibbs states for the field $\phi_x$.  When $d=2$, one must instead view the infinite volume measure as existing on the space of gradients.
If $\epsilon > 0$ and if $d \leq 4$, the second term  on the RHS grows with $\Lambda$ while, when $d=2$, it grows even after taking gradients in both arguments $x, y$. 

\item \textbf{Order-by-Disorder.}
On the other hand, the phenomenon of "Order-by-Disorder" provides related examples of systems which exhibit ordering due to various types of fluctuations.  The most relevant example, first considered by Henley \cite{Hen}, concerns a model Hamiltonian on $\Z^2$ of the form
\[
-\HH(\sigma )= \sum_{\|x-y\|_2=1} J_1 [\sigma_x- \sigma_y]^2 + \sum_{\|x-y\|_2^2} J_2 [\sigma_x- \sigma_y]^2\]
with $|J_1|< 2J_2$.  The  ground-states for this (frustrated) system are obtained by choosing a purely anti-aligned configuration of spins on each of the even and odd sub-lattices of $\Z^2$ and are thus parameterized by two angles:  an angle between the spin at $(0, 0)$ and the $e_1$-axis and relative angle between the spin at $(0, 1)$ and $(0, 0)$.  The degeneracy of ground-states is partially lifted under the  introduction of two types of "disorder".  The first type consists in passing from $0$ to positive temperature, see \cite{BCK} for mathematical justification of the effect in this case.  More relevant to the RFIO is a second mechanism: site dilution. Vertices of $\Z^2$ are deleted from the system independently with probability $p\ll 1$.  According to the calculations in \cite{Hen}, at $0$  and low temperature the system prefers the ground states with \textit{relative} angle between spins at $(0, 0)$ and $(0, 1)$ to be fixed at  $\pm \ffrac \pi2$.  

Besides the  obvious differences between this setup and ours, there is one crucial technical difference.  In the site diluted model, large fluctuations due to randomness are substantially weaker than those encountered in the analysis of the RFO(2)  model.  In particular, there is an analog to the field $g_x$ introduced below but the fluctuations of this field are about as singular as the four dimensional version of the RFO(2) model.  While the site diluted order-by-disorder problem has not been rigorously addressed, this feature suggests the conclusions in \cite{Hen} are reliable. It would be interesting to see if our methods can be adapted to this case.  \end{itemize}

Of the first three examples, the most worrying from the perspective of proving magnetic ordering in the RFO(2) model is the last one:  If we represent spins via angular variables -- $\sigma_x=(\cos(\theta_x), \sin(\theta_x))$ -- and make an expansion of \cref{E:Ham0} around $\theta_x=0$ (or any fixed angle $\psi$ for that matter) our model looks exactly like the random field Gaussian model.  This suggests that the whole ansatz of having order may be flawed since fluctuations in the latter model are so strong.

Beginning in the early 1980's, RFIO was the subject of a number of publications in the physics literature,  see \cite{DF1, DF2, MP}.  One group \cite{DF1} concluded there is a low temperature paramagnetic phase  The other \cite{MP} concluded there is an intermediate-temperature ordered phase from which they extrapolate the low-temperature behavior.  
Interesting tangentially related rigorous work was done in the 1990's on ground states in the strong field regime in \cite{Feld1, Feld2}.

The question of ordering in the RFO(2) model has also appeared in recent literature.  For Bose-Einstein condensates in optical traps, the effect was suggested as a possible response to the presence of certain kinds of experimentally realizable disorder \cite{Wehr-et-al-2, SPL-Nature}.  Here the (pseudo-)spin variables arise from internal structure of the atoms, the tuning of interactions and the structure of the optical lattice.  This type of phenomenon appeared as a possible mechanism for the splitting of Landau level degeneracy in graphene \cite{ALL}.  
Finally,  van Enter and coauthors \cite{vE1, vE2, vE3} came to this question during investigations of whether the spatial Markov property which characterizes Gibbs measures is preserved under various coarse-graining procedures.  

Little mathematically rigorous progress has occurred for RFIO except regarding qualitative ground state behavior and the mean field approximation \cite{vE1, vE2, vE3, Wehr-et-al-1, Wehr-et-al-2}.
Most recently, the author derived results \cite{NC} consistent with the picture presented above for the RFO(2) model with a Kac potential.  The fundamental limitation of that work is that the results are only valid if the range of the interaction potential is taken to diverge as a polynomial in $\epsilon^{-1}$.  This limitation is not entirely technical as we show below.  Even the nearest neighbor RFO(2) model has a fundamental length scale which is set by the behavior of energetics.

\subsection{Main Result}
The main theorem we aim to prove is the following:
Let $\epsilon$ be fixed,
\[
L \sim
\epsilon^{-1} \log^{4} \epsilon
\]
 so that $L=2^{k}$ for some $k \in \N$, and 
\[
Q_{L}(z)=
 z + \{0, 1 \dotsc, L-1\}^d.
\]
A subset $\Lambda$ of
$\Z^d$ will be said to be $L$-measurable if $\Lambda$ is a
union of blocks $Q_{L}(r)$ so that $r \in L\Z^d$.
We
define block average magnetizations \begin{align*} &
M_z =
\frac{1}{|Q_{L}|}\sum_{ x \in Q_L(z)}
\sigma_x,
\end{align*}

\begin{theorem}
\label{T:Main2} Let $d = 3$.
For any $\xi\in (0, 1)$ small enough and any $\epsilon \in (0,\epsilon_0(\xi))$ there is
$\beta_0(\epsilon, \xi)$ so that if $\beta> \beta_0$  the following holds:
Suppose that $\Lambda_N$ is a sequence of $L$-measurable volumes which increase to $\Z^3$ in a van Hove sense.    For almost
every $\omega\in \Omega$, there exists an
$L$-measurable subset $\D_{\omega} \subset \Z^d$ and
an $N_0(\omega)\in \N$ such that:
\begin{enumerate}
\item
$ |\D_{\omega}
\cap \Lambda_N| \leq
Ce^{-c|\log \epsilon|^{2}} |\Lambda_N|$ \text{ if $d=3$}, for all $N \geq N_0(\omega)$.

\item
For each $z\in \Lambda_N$ with $Q_{L}(z) \cap \D_\omega=
\varnothing$
and
\[ \|\langle M_z
\rangle_N^{\omega, e_1} - e_1\|_2 \leq \xi.
\]
Here $C, c$ are universal constants.
\end{enumerate}
\end{theorem}
Looking into the proof, we have a bound on the transition temperature of $\beta(\epsilon)\ls \epsilon^{-2}Poly(|\log \epsilon|)$ (we get this sort of bound also in the two dimensional case).  This is inline with the predictions of mean field theory, see for example \cite{C}.  On the other hand, the analysis in \cite{MP} suggests a random field strength independent transition temperature in the two dimensional case.  Whether this prediction is correct and what might take its place if it is not correct is unclear at present.  
\noindent
\textbf{Extensions.}
We have already mentioned that \cref{T:Main2} may be extended to include $d=2$ and also other (constant) boundary conditions.  We emphasize in the latter case that the derived conclusions will be the same as in \cref{T:Main2}:  if the constant boundary condition $u$ satisfies $u\cdot e_1>0$, then deep inside the volume $M_z$ will still be very close to $e_1$ .  Another direction in which the theorem can be extended is to incorporate spins $\sigma_x \in \mathbb S^n$.  In this case we take the disorder $\alpha_x$to be $n-1$ dimensional i.i.d. standard Gaussian vectors supported on  some hyperplane through the origin.  Ordering then occurs perpendicular to this hyperplane.
Finally, Gaussian disorder is taken only for convenience and \cref{T:Main2} can be extended to incorporate any choice of disorder for which the estimates of  \cref{L:RandBasic} (1)--(4) hold (Statement (5) of that lemma follows from Statements (1) and (3)).

\subsection{Some Intuition and an Outline of the Paper}
\label{S:Intuit}
In this subsection we give the heuristics and basic method of proof which lead to our result. The discussion will be given here for both $d\in\{2, 3\}$.  Here $\delta_{2, d}$ is the Kronecker $\delta$; it is $1$ if $d=2$ and $0$ otherwise. This explanation is a bit impressionistic to avoid some technical details.

First, let us indicate why one should expect ordering to occur in the horizontal direction (if there is to be ordering at all). 
We restrict attention to boxes  $Q_{\ell} \subset \Lambda_N$ of sidelength $\ell \leq  \epsilon^{-1}|\log \epsilon|^{-\frac{\delta_{2, d}}{2}}$ for numerous reasons which we hope will become apparent as the discussion proceeds.  

Let $\sigma_x$ be given in polar coordinates by an angle $\theta_x$.  We perform a "spin wave" analysis supposing that in $Q_{\ell}$ $\theta_x = \psi+ \hat{\theta}_x$ with $\hat{\theta}_x$ small and $\psi$ fixed.   Then we expand $-\HH_{Q_\ell}$ (with free boundary conditions) in $\hat{\theta}_x$ variables.  Keeping only terms up to second order in $\hat{\theta}_x$ and $\epsilon$ we find
\[
\sup_{\stackrel{(\theta_x)_{x \in Q_{\ell}}}{\theta_x \approx \psi}} -\HH_{Q_\ell}(\theta) = -\frac{\epsilon^2}{2}\cos^2(\psi)\sum_{x \in Q_{\ell}} \hat{\alpha}_x \Delta^{-1} \cdot \hat{\alpha}_x + \underbrace{O(\epsilon|\sum_{z \in Q_{\ell}} \alpha_z|)}_{\textrm{I}}.
\]
where $\Delta$ is the Neumann Laplacian for $Q_{\ell}$ and $\hat{\alpha}_x= \alpha_x - |Q_{\ell}|^{-1} \sum_{z\in Q_{\ell}} \alpha_z$.
The directions of presumed ordering are obtained by optimizing the first term in $\psi$ and ignoring Term \textrm{I}.
An important additional observation is that the optimal choice for the deviation variables $\hat{\theta}_x$ is $\hat{\theta}_x=  \cos(\psi) g^N_x$ where 
\[
g^N_x:=-\epsilon \Delta^{-1} \hat{\alpha}_x.
\]

In evaluating the validity of this computation we run into some constraints on the box size $\ell$.
Typically,
\begin{equation}
\label{E:22}
\sum_{x \in Q_{\ell}} \hat{\alpha}_x \Delta^{-1} \cdot \hat{\alpha}_x \sim \epsilon^2  \ell^d\log^{\delta_{2, d}} \ell.
\end{equation}
So, for $d\geq 3$, the central limit theorem implies that Term \textrm{I} will be of lower order (with high probability) if $\ell\gg \epsilon^{-\frac 2d}$.  
For $d=2$ there is crucially an extra factor of $\log \ell$ in \cref{E:22}.   This allows us, for appropriate choice of $\gamma < \ffrac 12 $, to suppress Term \textrm{I} for $\ell \gg \epsilon^{-1}|\log \epsilon|^{-\frac 12-\gamma}$ with probability exponentially small in $-|\log \epsilon|^{1-2\gamma}$.   Together these considerations give a lower bound on $\ell$ for $d \geq 2$.

On the other hand, in low dimension the behavior of the field $g^N_x$ provides an \textit{upper bound} on $\ell$.  This field has a typical order of magnitude in dimensions $d=2, 3, 4$ of $\ell, \sqrt{\ell}, \sqrt{\log \ell}$ respectively.  Thus the calculation cannot be taken too seriously for large boxes as in this case the maximizer is inconsistent with starting assumption that $\hat{\theta}_x$ is (uniformly) small.  In two and three dimensions, we arrive at the following constraints:
\begin{align}
& \epsilon^{-1}|\log \epsilon|^{-\frac 12-\gamma}\ll \ell \ll  \epsilon^{-1}|\log \epsilon|^{-\frac 12} \quad \text{if $d=2$},\\
& \epsilon^{-\frac 23}\ll \ell \ll \epsilon^{-2}|\log \epsilon|^{-\frac 12} \quad \text{ if $d=3$.}
\end{align}
where the $|\log \epsilon|^{-\frac 12}$ factor on the right hand bound was inserted to account for possible fluctuations of $\|g^N\|_{\infty}$.  This indicates the marginal nature of the two dimensional model.

From a technical perspective, the spin wave calculation seems to depend on uniform control of $\hat{\theta}_x$.  Unfortunately, the only control we have, stated in  \cref{L:FBE}, concerns the Dirichlet energy of spin configurations in $Q_{\ell}$,  $\EE_{Q_\ell}(\sigma):= \sum_{\langle x, y\rangle \subset Q_{\ell}} [\sigma_x - \sigma_y]^2$.  It is energetically favorable for
 \begin{equation}
 \label{E:DE}
 \EE_{Q_\ell}(\sigma) \ls  \epsilon^2 |\log \epsilon|^{\delta_{2,d}} \ell^d
 \end{equation}
 but no better.
Indeed, from the spin wave heuristic our best guess on the maximizer of $-\HH_{Q_\ell}$ is $g_x^N$ which has Dirichlet energy on this order.

Because we can only control the oscillations of $\sigma$ in $H^1$, we needed a different computational device which gives the same conclusions as the spin wave analysis.  This is presented in detail in \cref{L:BLayer}. The physical idea which sits in the background is that low energy spin configurations in the RFO(2) model consist of two superposed modes.   In angular variables, with $\sigma_x$ replaced by the angle $\theta_x$, $\theta_x$ consists of a slowly varying configuration $\phi_x$, whose Dirichlet energy is negligibly small, and a quickly varying mode which is proportional to $g^N_x$.  A convenient way to distinguish between these two modes is to define $\phi_x$ by
\begin{equation}
\label{E:COVintro}
\phi_x = \theta_x-\cos(\theta_x)g^N_x.
\end{equation}
If we regard this as a change of variable $\theta_x \mapsto \phi_x$ the Hamiltonian $-\HH_{Q_\ell}$ transforms into a new energy functional
\[
\KK_{Q_\ell}(\phi_x)= \underbrace{\sum_{\langle xy\rangle} \cos(\phi_x- \phi_y)-1}_{\textrm{II}} + \frac 14 \sum_{x} m_x \cos^2(\phi_x)
\]
with
\[
m_x= \sum_{y: |x-y|=1} [g^N_y-g^N_x]^2.
\]
There are errors made in this transformation but we can  control these errors using $\EE_{Q_\ell}(\sigma)$ and the behavior of the two fields $g^N_x$ and $\alpha_x$ in $Q_\ell$ so long as $\ell$ is neither too small or too big as indicated above.  Note that the potential $m_x \sim \epsilon^2 |\log \epsilon|^{\delta_{2, d}}$ typically. Thus if $\phi_x$ only varies at a much longer length scale than  $g^N_x$ we may pretend that $m_x$ has been replaced by this constant.  Because Term \textrm{II} punishes configurations for which $\phi_x$ varies quickly, the conclusions of the spin wave computation are still valid.  

Using the above analysis, we next setup a coarse-grained Peierls argument using the scale $\ell$ as the microscopic scale and show how to extract energetic cost from the occurrence of contours.  The technical problem in implementing this idea is that the spin space is continuous and there is no microscopic surface tension.  This is analogous to the problem encountered in constructing a Peierls estimate for a "soft" Ising model in which continuous scalar spins sit in a double well potential.  Limitations inherent in the bound \cref{E:DE} make the present analysis more challenging.  To solve this problem, we take inspiration from previous work on Kac models, see  \cite{Pres-Book} and also \cite{NC}.  This method requires us to use two scales;  the scale $\ell$ and a second scale $L\gg\ell$ with contours defined relative to the second scale $L$

Before describing contours let us fix the scales $\ell, L$.  A bit later we will see that it is important that $\ell \ll \epsilon^{-1}| \log \epsilon|^{-\ffrac {\delta_{2, d}}2}$.  We take
\begin{align*}
&\ell\sim \epsilon^{-1} |\log \epsilon|^{-4} \quad \text{ and } \quad L\sim \epsilon^{-1} |\log \epsilon|^{4} \text{ if $d=3$},\text{ if $d=3$},\\
&\ell\sim \epsilon^{-1} |\log \epsilon|^{-\ffrac 12-\ffrac 1{128}} \quad \text{ and } \quad L\sim \epsilon^{-1} |\log \epsilon|^{-\ffrac 12+\ffrac 1{128}} \text{ if $d=2$}.
\end{align*}

To define contours, we first classify cubes $Q_{\ell}$ as either good or bad relative to a spin configuration $\sigma$.  Fixing $\sigma$, call a box $Q_\ell$ of side-length $\ell$ \textit{bad} for $\sigma$ if either  the Dirichlet energy $\EE_{Q_\ell}(\sigma)$
is substantially larger than $4\epsilon^2  \ell^d \log^{\delta_{2, d}} \ell$
or if the average of spins in $Q_\ell$, $\sigma(Q_\ell):=\ell^{-d} \sum_{x \in Q_{\ell}} \sigma_x$, is far from $\pm e_1$; see \cref{S:Contours} for further details.
A semi-precise definition of a contour $\Gamma$ is as follows: contours are maximally connected clusters of boxes $Q_L$ of side length $L$  so that within distance $2L$ of  $Q_L$ there is a cube $Q_\ell$ of side-length $\ell$ which is bad for $\sigma$. 

The reason one introduces the second scale $L$ is that it allows us to do surgery at the boundary of a contour to artificially create surface tension.  Given a spin configuration $\sigma$ and an associated contour $\Gamma$ let $\delta(\Gamma)$ denote the neighborhood of radius $L$ around $\Gamma$.
 The idea is to compare $\sigma$ with a new spin configuration $\tilde{\sigma}$ which agrees with either $\sigma$ or the reflection of $\sigma$ across the $e_2$ axis on each component of $\Lambda \backslash \delta(\Gamma)$. We construct $\tilde \sigma$ to have an angle uniformly close to either $0$ or $\pi$ on the whole of $\Gamma$.  The skeleton of this argument is given in \cref{S:Groundstates,,S:Collar,,S:Peierls}, proofs of lemmas from these sections appear in \cref{SS:6,,SS:7,,S:PP}.

The key, and most labor intensive, part of the construction requires us to do surgery on the spin configuration $\sigma$ at the boundary of the contour, i.e. on $\delta(\Gamma) \backslash \Gamma$.
We will assume the randomness is "uniformly typical" in the discussion of this procedure.  Allowances must be made in the actual analysis and precise conditions we require of the randomness on a contour $\Gamma$ are detailed in \cref{S:DisType}. Not all contours satisfy these conditions, but we state lemmas which quantify the size and sparsity of regions where our requirements fail.  The probabilistic bounds required to prove these lemmas and the proofs of the lemmas themselves are postponed until \cref{S:Randomness}.
  
By definition, $\sigma$ is good in $\delta(\Gamma)\backslash \Gamma$: for each cube $Q_\ell$ within $L$ of $\delta(\Gamma)\backslash \Gamma$
\[
\EE_{Q_\ell}(\sigma)\leq 4 \epsilon^2 \ell^d \log^{\delta_{2, d}} \ell,
\]
and $\ell^{-d} \sum_{x\in Q_\ell} \sigma_x$ is close to $\pm e_1$.  Combining these to facts, one can show that the sign $\pm1$ is constant over connected components of $\delta(\Gamma)\backslash \Gamma$.  

Fix a component $R$ of $\delta(\Gamma)\backslash \Gamma$ and suppose $\ell^{-d} \sum_{x\in Q_\ell} \sigma_x$ is close to $e_1$ throughout $R$ (the other case is similar).  The reason our analysis is restricted to low dimension is that if $d\in \{2, 3\}$, the \textit{ a priori } bounds on $\EE_{Q_\ell}(\sigma)$ allow us to find a small neighborhood $R' \supset R$ so that for vertices $x$ in the graph boundary of $R'$, $\theta_x \in (-\ffrac \pi8, \ffrac \pi8)$.  This point is addressed in \cref{S:AuxBound}.

Next, if we start with such a boundary condition $\sigma_0$ for $R'$, we can look for maximizers of $-\HH_{R'}(\sigma|\sigma_0)$.  Actually, it is better to use \Cref{L:BLayer} and look for maximizers after a change of variables. 
The idea is the same as in \cref{E:COVintro} but we use the Dirichlet Laplacian $\Delta_{R'}^D$ in $R'$ instead of the Neumann Laplacian.
Set
\[
g^D_x=\epsilon[-\Delta^D_{R'}+ \lambda]^{-1} \alpha_x
\]
for all $x$ in $R'$.  The mass $\lambda= L^{-2} \log ^8 L$ if $d=3$.  It is a technical convenience used to keep $g^D_x$ small in absolute value throughout $R'$ (this is a point where one has to make modifications in two dimensions).
We make the change of variables
\begin{equation}
\label{E:COV2}
\theta_x \mapsto \theta_x- \cos(\theta_x) g^D_{R,x} =: \phi_x
\end{equation}
Because of the Dirichlet boundary conditions the restriction of $\sigma$ to ${R'}^c$, $\sigma|_{{R'}^c}$, is invariant under this transformation.
On the other hand, $-\HH_{R'}(\sigma|\sigma_0)$ transforms into 
\[
\KK_{R'}(\phi_x|\phi_{{R'}^c})= \sum_{\langle xy\rangle\cap R' \neq \varnothing} \cos(\phi_x- \phi_y)-1 + \frac 14 \sum_{x} m_x \cos^2(\phi_x)
\]
where
\[
m_x= \sum_{y: |x-y|=1} [g^D_y-g^D_x]^2.
\]
This change of variables potentially costs us energetically, but we will more than offset any loss by what we gain in comparing $\sigma$ to $\tilde{\sigma}$ (which is still to be constructed).

One can show that if the randomness is well enough behaved in $R'$, there is a unique maximizer of $\KK_{R'}$ and that this maximizer also has $\phi_x \in (-\pi/8, \pi/8)$ throughout $R'$.  Because of this and the fact that typically $m_x \sim \epsilon^2 |\log \epsilon|^{\delta_{d, 2}}$, one should morally regard the maximizer as behaving like a solution to the discrete elliptic PDE with mass 
\begin{equation}
\label{E:EPDE}
-\Delta \phi_x+ \epsilon^2 |\log \epsilon|^{\delta_{2, d}} \phi_x=0.
\end{equation}
This leads us to the separation of length scales at $\epsilon^{-1}|\log \epsilon|^{-\frac{\delta_{2, d}}{2}}$ which we chose above.  A solution to \cref{E:EPDE} is very close to $0$ deep inside $R'$ only if the inner radius of $R'$ is  much larger than $\epsilon^{-1}|\log \epsilon|^{-\frac{\delta_{2, d}}{2}}$.  This issue is addressed in \cref{S:AuxRelax}.  On the other hand, in order to distinguish between the pure phases $\pm e_1$, we need $\ell$ to be smaller than this scale since $\epsilon^{-1}|\log \epsilon|^{-\frac{\delta_{2, d}}{2}}$ will be the typical thickness of interfaces between pure phases.

After replacing $\sigma$ by the maximizer of $\KK_{R'}$ inside $R'$ for each component $R$ of $\delta(\Gamma) \backslash \Gamma$, we can invert the change of variables \cref{E:COV2} to produce a new spin configuration $\sigma^1$, which disagrees with $\sigma$ only on $\delta(\Gamma) \backslash \Gamma$.
Because $\sigma^1$ is very close to one of $\pm e_1$ deep inside each thickened boundary component of $\Gamma$, at negligible further cost we can force it to be \textit{exactly} $\pm e_1$ in the middle of each of the boundary components, calling the result $\sigma^2$.  

Next we modify $\sigma^2$ into a new configuration $\sigma^3$ on $\Lambda_N\backslash \Gamma$.  We keep $\sigma^2$ fixed on the exterior component $\mathscr E$ of $\Lambda_N\backslash \Gamma$ and note the value $\sgn(\sigma_2\cdot e_1)$ on $\mathscr E \cap \delta(\Gamma) \backslash \Gamma$. We make the following transformation on the interior components $\II_j$.  If on $\II_j \cap \delta(\Gamma) \backslash \Gamma$  $\sgn(\sigma_2\cdot e_1)$ agrees with its value on the boundary of the exterior component, we keep $\sigma_2$ fixed on the entire interior component.  Otherwise, we reflect $\sigma_2$ across the $e_2$ axis on all of $\II_j$.  We call the result, which is defined only in $\Lambda_N\backslash \Gamma$, $\sigma^3$. 

Let $\bar{\delta}(\Gamma)$ be the neighborhood of radius $\ffrac L2$ of $\Gamma$.  To obtain $\tilde{\sigma}$ from $\sigma^3$, we construct (almost-)maximizers for $-\HH_{\bar{\delta}(\Gamma)}(\cdot)$ (here we use free boundary conditions).  This gives two configurations $\sigma^{\pm}_{\Gamma}$ uniformly close  to $\pm e_1$ throughout $\bar{\delta}(\Gamma)$.  The configuration  $\tilde{\sigma}$ is chosen to agree with $\sigma^3$ on $\bar{\delta}(\Gamma)^c$ and to agree with the appropriate maximizer of $-\HH_{\bar{\delta}(\Gamma)}(\cdot)$ inside $\bar{\delta}\Gamma$.  After accounting for all errors and also what we gain in the comparison between $\sigma$ and $\tilde{\sigma}$ we are able to show
\[
-\HH_{\Lambda_N}(\tilde{\sigma}|e_1)+\HH_{\Lambda_N}(\sigma| e_1) \geq C \epsilon^2|\log \epsilon|^{-24}|\Gamma|
\]
for some $C>0$.
With this result in hand we can construct a Peierls argument and prove \cref{T:Main2}.

The remainder of this paper is organized as follows.  In \cref{S:Prelims}  we collect basic notation used throughout the rest of the article.  \Cref{S:Contours} defines what are contours in our model in a careful way, while \cref{S:DisType} addresses what sort of requirements we need from the randomness $\alpha_x$ in order for a subregion of $\Z^3$ to be well behaved.  Regions satisfying these requirements are called clean regions and in \cref{P:Dirt4,,L:BadSetBound} we quantify the likelihood and sparsity of dirty regions.  Proofs related to \cref{S:DisType} are postponed to \cref{S:Randomness}. In \cref{S:4} we state the Peierls estimate, \cref{L:Contours}, and prove \cref{T:Main2} on its basis.  Assuming that a given contour is clean, \cref{S:Groundstates} outlines the construction of almost maximizers to $-\HH_{\bar{\delta}(\Gamma)}(\cdot)$ while \cref{S:Collar} outlines the construction of $\sigma_1$ above.   In \cref{S:Peierls} we construct $\tilde{\sigma}$ and then prove the Peierls estimate.  Proofs of major estimates required in \cref{S:Groundstates,,S:Collar,,S:Peierls} are postponed until \cref{SS:6,,SS:7,,S:PP} respectively.  \Cref{S:Energetics}, while technical, is key to the whole paper as it contains \cref{L:BLayer}.  Finally \cref{S:Aux} contains various technical lemmas needed elsewhere, for example in the analysis of maximizers of  $\KK_{R'}(\phi_x|\phi_{{R'}^c})$.

\section{Preliminaries}
\label{S:Prelims}
Throughout $\sigma_x \in \mathbb S^1$.  We will use the notation $e_1, e_2$ for the usual orthonormal basis of $\R^2$ in which $\mathbb S^1$ sits.  We will also use $\{\hat e_1, \dotsc, \hat e_d\}$ to denote the standard basis in $\Z^d$ and $\R^d$, the $\hat{\cdot}$ being used to distinguish unit vectors which live in the spatial lattice from those which live in spin space.
In general $\|\cdot\|_p$ denotes the $\ell^p$ norm in $\R^d$. If $R \subset \Z^d$ and $f: R \rightarrow \R^k$, we let
\[
\|f\|_{p, R}= \left[\sum_{x\in R} \|f\|_2^p\right]^{1/p}.
\]

We write $ x\sim y$ $x, y \in \Z^d$ if $\|x- y\|_1=1$.
For any finite subset $R \subset \Z^d$,
let
\begin{align*}
\pI R& =\{x \in R: \dist_1(x, R^c) \leq 1\}\\
 \pO R&=\{x \in R^c: \dist_1(x, R) \leq 1\}
\end{align*}
with $\dist_p(A, B)$ the Hausdorff distance between to sets in $\ell^p$. $|R|$ will denote (depending on the context) the cardinality (resp. volume) of a finite set $R\subset \Z^d$ (resp. bounded domain in $\R^d$).

The proof requires that we work on multiple scales.  This is conveniently implemented by considering $L' \in (2^{k})_{k \in \N}$.
We will use the notation
\[
Q_{L'}(x_0)=\{z \in Z^d: z= x_0+ v \text{ for some $v \in \{0, \dotsc, L'-1\}^d$}\}
\]
to denote an arbitrary box of side-length $L'$ indexed (in the "lower left corner") by some fixed $x_0 \in \Z^d$.  Also, let $B_{L'}(x_0)= x_0+ \{-L'/2, \dotsc, L'/2-1\}^d$ denote the cube with side-length $L'$ centered (roughly) at $x_0$.
We shorthand these by suppressing reference to $x_0$, writing $Q_{L'}, B_{L'}$ instead when no confusion should arise.

Let $\alpha_x, x \in \Z^d$ be a field of i.i.d. Gaussian random variables with mean $0$ and variance $1$. We treat in full only the case of Gaussian disorder for now.  Our present proof should work for any random variables with subGaussian tails by extending the estimates of \cref{S:Randomness}.  To what extent we need subGaussian tails, and not just exponential moments in a neighborhood of $0$ is not clear.

We let $\mathcal S_{\Lambda}=[\mathbb S^1]^{\Lambda}$ for $\Lambda \subset \Z^d$ denote the configuration space for a spin system restricted to $\Lambda$, with $\mathcal S=\mathcal S_{\Z^d}$. If $\Lambda_1 \subset \Lambda_2$ and $\sigma \in \mathcal S_{\Lambda_2}$ then $\sigma|_{\Lambda_1} \in \mathcal S_{\Lambda_1}$ denotes the natural restriction to $\sigma$ to $\Lambda_1$.
Let
\[
\mathcal E_R(\sigma)= \sum_{\stackrel{\langle xy \rangle}{x, y \in R}} (\sigma_x- \sigma_y)^2
\]

Given $\sigma_0 \in \mathcal S$ and $\sigma\in \mathcal S_R$ we may extend $\sigma$ to an element of $\mathcal S$ via $\sigma|_{R^c}= \sigma_0|_{R^c}$.  In this case we let
\begin{equation}
\label{L:Ham}
-\HH_{R}(\sigma|\sigma_0)= -\frac 12 \mathcal E_R(\sigma) - \frac 12\sum_{\stackrel{x \in \pI R, y \in \pO R}{x \sim y}} (\sigma_x-\sigma_y)^2+ \epsilon \sum_{x \in R} \alpha_x e_2 \cdot \sigma_x.
\end{equation}
We sometimes consider regions with free boundary conditions, in which case the Hamiltonian is denoted by $-\HH_{R}(\sigma)$ and reads the same as  \cref{L:Ham} except that the second term on the RHS has been dropped.

The edges $e=\langle x, y \rangle$ come with a natural orientation: we will say that $e$ is positively oriented if the vector $(y-x) \cdot \hat{e}_i \geq 0$  for all $i\in \{1, \dotsc, d\}$.  Given a vector valued function $f: \Z^d \rightarrow \R^k$ we can associate a discrete vector field ( a function on edges) via
\[
\nabla_e f= f_{y}-f_x
\]
whenever $e= \langle x, y \rangle$ is positively oriented.  
In a region $R$ and given $f :R \rightarrow \R^k$, we introduce the two discrete Laplacian operators by $\Delta^D_R, \Delta^N_R$ (Dirchlet and Neumann) by
\begin{align*}
&-\Delta^D_R\cdot f_x = \sum_{\stackrel{y \sim x}{y \in R \cup \pO R}} f_x-f_y \text{ where $f$ is extended to be $0$ on $\pO R$},\\  
&-\Delta^N_R\cdot f_x= \sum_{\stackrel{y \sim x}{x, y \in R}}  f_x-f_y.
\end{align*}
As a general notation, if $f: R \rightarrow \R^k$ is any (vector valued) function,
\begin{align*}
&f(R):= \frac 1{|R|}\sum_{x\in R} f_x .\\
&\hat{f}_x= f_x-f(R).
\end{align*}
For any $\lambda \geq 0$ set
\begin{align*}
&\gld_x= g^{\lambda, D}_{x, R} = \epsilon (-\Delta^D_R+ \lambda)^{-1}\cdot \alpha_x,\\
&\gln_x= g^{\lambda, N}_{x, R} = \epsilon (-\Delta^N_R+ \lambda)^{-1} \cdot\hat{\alpha}_x.
\end{align*}

\noindent
\textbf{Calculational Convention:}
In our estimates, we will use the constants $C, c$ to denote generic universal constants whose values change from line to line.  
At various points in the work we use $O(\cdot)$ notation:  An expression $g \in O(f)$ for some other expression $f$ if there exists a \textit{universal} constant $C>0$ so that
\[
|g| \leq C f.
\]
We will also say
\[
g \ls f
\]
if there is a universal constant $C>0$ so that $g \leq Cf$.

\section{Course-Graining, Contours, and Disorder Types} \label{S:Contours}
We introduce two length scales $\ell, L \in  (2^{k})_{k \in \N}$ where
\begin{equation}
\label{e:scales}
\begin{split}
\log_2(\ell)&=\lfloor\log_2\left(\epsilon^{-1} |\log \epsilon|^{-4}\right)\rfloor\\
\log_2(L)&= \lceil \log_2\left(\epsilon^{-1} |\log\epsilon|^{4}\right)\rceil.
\end{split}
\end{equation}
The scales $\ell$ and  $L$ introduce coarse-grainings of $\R^d$ and, by taking intersections, of $\Z^d$.  For any $L_0 \in \N$, we shall
say that a block $Q_{L_0}(r)$ is measurable relative to the scale $L_0\in \N$ if $r \in L_0 \Z^d$.  Denote this standard collection of blocks by $\QQ_{L_0}$ and set $\NN:=\ffrac{L_0}{16}\{-32,\dotsc 32,\}^d$.   For a fixed $L_0$-measurable box $Q_{L_0}(r)$ and any $\eta \in  \NN$ let $Q_{\eta}:= Q_{L_0}(r + \eta)$.  Set $\QQ_{L_0}'= \{ Q_{\eta}: Q\in \QQ_{L_0} \text{ and } \eta \in N\}$.
Finally, $\QQ^s_{L_0}=\{Q : Q-(\ffrac{L_0}{2}, \dotsc, \ffrac{L_0}{2}) \in \QQ_{L_0}\}.$

For any set $A \subset \Z^d$ we can associate a subset $\hat A
\subset \R^d$ as the union of closed boxes of side length $1$ centered
at the elements of $A$.  We shall say that $A$ is connected if
$\hat{A}$ is.  Note this is NOT the same as connectivity in terms of the graph structure on $\Z^d$.  We will refer to connectivity in the latter sense as \textit{graph connectivity}.
$\hat{A}^c$ decomposes into one infinite connected component
$\Ext(\hat{A})$ and a number of finite connected components
$(\Int_i(\hat{A}))_{i=1}^{m}$.  The union of finite components is denoted by $\Int(\hat{A}):= \cup_{i=1}^m
\Int_i(\hat{A})$. Let us denote the $L_0$-enlargement of a set $\hat A
\subset \R^d$ by
\begin{equation}
\label{E:Sets}
\delta_{L_0}(\hat{A})=\cup_{\{Q_{L_0}(r) \:
L_0\text{-measurable}\: :\: \dist(\hat Q_{L_0}(r), \hat{A})<
\ell\}}\hat Q_{L_0}(r)
\end{equation}
where $\dist(\hat Q_{L_0}(r), \hat{A})$ is the
Hausdorff Distance between sets in $\R^d$ in the $\ell_{\infty}$
metric. We also introduce $\delta_{L_0}(A), \Int_i(A), etc.$ by
taking intersection of each defined set with $\Z^d$. Most commonly, we will use this notation with $L_0=L$ in which case we suppress the subscript, writing $\delta(A)$
The closed hull of a
set $A \subset \Z^d$ is defined by $c(A):=\delta(A)\cup \Int(A)$ (with $L_0=L$).

Given a spin
configuration $\sigma \in \mathcal S$ we introduce the following
phase variables:
\begin{itemize}

\item
Let
\[\psi^{(0)}_z(\sigma)=
\begin{cases}
1 \quad \text{ if $\EE_{Q_\ell(r')}(\sigma) \leq  \epsilon^{2}|\log \epsilon| |Q_{\ell}|$}\\
\quad \:\:\text{ for all $r'$ such that $\dist(z,r') \leq 5\ell$.} \\
 0\quad \text{ otherwise}.
\end{cases}
\]
We expect (see \cref{S:Energetics}) that, typically,  low energy configurations have
\[
\EE_{Q_\ell(r')}(\sigma) \ls \epsilon^2 |Q_{\ell}|;
\]
the extra $|\log \epsilon|$ is taken here for convenience and is not optimal.  In two dimensions, we must be more careful.

\item
Let
\[
\psi^{(1), \xi}_z(\sigma)=
\begin{cases} 1 \quad \text{ if $\sigma(Q_\ell(r'))\cdot e_1 \in [1-\xi, 1]$ whenever $\dist(z,r') \leq 5\ell$},\\
-1\quad \text{ if $\sigma(Q_\ell(r'))\cdot e_1 \in [-1, -1+\xi]$ whenever $\dist(z,r') \leq 5\ell$},\\
0 \quad \text{ otherwise.}
\end{cases}
\]

\item
Let
\[
\psi_z=\psi^{(1)}_z\psi^{(0)}_z.
\]
\end{itemize}
We will often suppress sub/superscripts and the argument $\sigma$ from these functions, writing $\psi^{(0)}, \psi^{(1)}, \psi$.

\noindent
These phase variables are extended to the larger length scale $L$ as follows.  For coarse-graining purposes we consider only boxes $Q_{L}(r)$ which are $L$-measurable.  If $z \in Q_L(r)$ with $Q_L(r)$ an $L$-measurable box, we set
\[
\Psi^{\xi}_{z}(\sigma)=
\begin{cases}
1 \quad
\text{ if $\psi^{(1)}_y\psi^{(0)}_y=1$ for all  $y \in Q_{L}(r')$, where $Q_{L}(r') \in \QQ_{L}$}\\
\quad \quad  \text{and $\|r-r'\|_{\infty} \leq 2L$},\\
-1\quad \text{ if $\psi^{(1)}_y\psi^{(0)}_y=-1$ for all  $y \in Q_L({r'})$, where $Q_{L}(r') \in \QQ_{L}$}\\
 \quad \quad  \text{and if $\|r-r'\|_{\infty} \leq 2L$},\\
 0 \quad \text{ otherwise}.
\end{cases}
\]
By definition, $\Psi$ is an $L$-measurable piece-wise constant function on $\Lambda$.

Given $\sigma \in \mathcal S$, it is a standard argument that the connected components of
\[
R^+= \{z \in \Z^d:
\Psi^{\xi}_z=1\}\]
 and
\[R^-= \{z\in \Z^d : \Psi^{\xi}_z=-1\}
\]
are
separated by connected subsets of
\[
R^0= \{z\in \Z^d:
\Psi^{\xi}_z=0\}.
\]

\begin{definition}
\label{D:Contour}
An \textit{abstract contour} $\Gamma$ is defined to be the
pair $(\Sp(\Gamma), \psi_{\Gamma})$ where $\Sp(\Gamma) \subset
\Z^d$ is connected and $L$-measurable and $\psi_{\Gamma}(z)$ is
a $\{-1, 0, 1\}$-valued function on
$\Sp(\Gamma)$ which gives the values of the phase specification on
$\Gamma$. In the definition surrounding \eqref{E:Sets}, whenever the set $A$ in
question happens to be $\Sp(\Gamma)$ we will write $\delta(\Gamma),
c(\Gamma),$ etc.   Let $N^L_{\Gamma}=
|\delta(\Gamma)|/L^d$, $N^\ell_{\Gamma}=|\delta(\Gamma)|/\ell^d$ so that $N^L_{\Gamma}, N^{\ell}_\Gamma \in \mathbb
N$.
\end{definition}
\begin{remark}[Other Choices of Boundary Conditions]
Contours as we have defined here are not adequate for dealing with boundary conditions which are not $\pm e_1$ because spins \textit{should} be close to a chosen boundary condition in boundary boxes.  One way to deal with this to redefine what it means to be good or bad according to how close a box is to the boundary of $\Lambda_N$.  In any case, the analysis of \cref{S:Collar} must be modified for boundary contours since we are not free to change the boundary conditions.
\end{remark}

Let us introduce the notation
$\delta_{ext}(\Gamma)=\delta(\Gamma)\cap \Ext(\Gamma)$ and
$\delta^{i}_{in}(\Gamma)= \delta(\Gamma) \cap \Int_i(\Gamma)$.  These sets are
evidently disjoint.  By definition of $\Gamma$, each of these sets is
connected.

Given an abstract contour $\Gamma$ we shall say that $\Gamma$ is a \textit{concrete contour}, or just a contour, for $\sigma$ if
$\Sp(\Gamma)$ \text{ is a maximal connected subset of }
$R^0(\sigma)\text{ and } \psi_z(\sigma) \equiv
\psi_{\Gamma}(z)\text{ on } \Sp(\Gamma)$.  It is possible that an abstract contour is never be realized concretely.

We shall denote by
\begin{multline} \mathbb X(\Gamma)= \{ \sigma\in
\Omega: \Gamma \text{ is a contour for $\sigma$}\}\\ := \{\sigma:
\Sp(\Gamma) \text{ is a maximal connected subset of }
R^0(\sigma)\text{ and } \psi_z(\sigma) \equiv
\theta_{\Gamma}(z)\text{ on } \Sp(\Gamma)\}. \end{multline}
Put another way, we shall
say that $\Gamma$ is a contour for $\sigma$ if $\sigma \in \mathbb
X(\Gamma)$.

From our definitions, specifying that $\Gamma$ is a contour of $\sigma$
lets us recover the values of $\Psi_z(\sigma), \psi_z(\sigma)$ on
$\delta(\Gamma)$ just from the function $\psi_{\Gamma}(z)$ (see \cite{Pres-Book} for details).
This convenient
property allows us analyze systems of \textit{compatible} contours without worrying about
the microscopic spin configuration far away from the contour. Two
contours $\Gamma_1, \Gamma_2$ are said to be compatible if
$\delta(\Gamma_1)\cap \Sp(\Gamma_2)= \varnothing$ and
$\psi_{\Gamma_1}= \psi_{\Gamma_2}$ on the domain of intersection
of $\delta(\Gamma_1), \delta(\Gamma_2)$.

So far we have considered contours at the level of spin
configurations.  We would like to be able to show that under certain
finite volume Gibbs measures, a contour costs $e^{-q|\Gamma|}$, the constant $q$ being made large by appropriate choice
of $\beta, \xi$ and $\epsilon$.  In general, such an estimate will \textit{never} be true
uniformly in the presence of randomness. However, it turns out that large fluctuations
of the fields, to be quantified in the next subsection, \cref{S:DisType}, and in \cref{S:Randomness}, don't effect
the extraction of energy from \textit{most} contours.

Thus we turn to the interplay between spin configurations and
randomness. We introduce (more) phase variables associated to the
randomness $\omega\in \Omega$.  For this we need some more definitions.

\subsection{Disorder}
\label{S:DisType}
Let $L_0$ be a fixed scale of order $\epsilon^{-1}$. Later we specialize to three scales: $\ffrac{\ell}2, \ell$ and $L$ as in \eqref{e:scales}.  
We will consider the behavior of the fields $\gll_{Q_{L_0}}$ for \text{both} Dirichlet and Neumann boundary conditions. The superscript distinguishing the boundary conditions is suppressed.

First we need a coarse estimate on the behavior of the fields $\gll_{Q}$.  For a given cube $Q$ we briefly introduce the potential
\[
m_{Q, x}:=\sum_{\stackrel{e\cap Q \neq \varnothing}{e\ni x }} [\nabla_e \gld_{Q}]^2
\]
which plays an important role in \cref{S:Groundstates,,S:Collar}.  We take $\lambda \in [0, 1)$ fixed. 
Let $A, B>0$ be fixed constants.  Consider the event
\[
\mathcal A_Q(r)=\\
\left\{\omega: r^{-d}\sum_{\|y-x\|_\infty \leq r} m_{Q, y} \geq A \epsilon^2  \text{ for all $x\in Q$ s.t. }{\dist_{\infty}(x, \pO Q) \geq \frac{L_0}{16} } \right\}.
\]
A box $Q\in \QQ_{L_0}'$, will be called \textit{nice} if for all $\lambda < L_0^{-1}$:
\begin{align}
&\label{E:r1}\text{The event $\AA_{Q_{\eta}}(r)$ occurs with $r= \log^{90} L_0$.}\\
&\label{E:r2} \|\gll_{Q_\eta}\|_{\infty} \leq 
\epsilon L_0^{-\frac 12} |\log \epsilon|^{30}\\
&\label{e:r3}
\|\nabla \gll_{Q_\eta}\|_{\infty} \leq 
 \epsilon |\log \epsilon|^{30}\\
& \label{E:r3}
A \epsilon^2 \leq \|\nabla \gll_{Q_\eta}\|_2^2/|Q_\eta| \leq  B \epsilon^2 \\
&\label{E:r4}
\|\alpha\|_{\infty, Q_\eta} \leq |\log \epsilon|^{30}.\\
&
|\alpha_{Q_{\eta}}|\sqrt{|Q_{\eta}|} \leq 
|\log \epsilon|^{30}.
\end{align}
A box $Q_{L_0}(r)\in \QQ_{L_0}$, will be called \textit{good} if $Q_{\eta}$ is nice for all $\eta \in\NN$.    A box $Q \in \QQ_{L_0}'$ is it called good if it intersects a good box in $ \QQ_{L_0}$.  We introduce the function on $\QQ_{L_0}' $ by
\[
\Xi_{L_0}(Q)=\begin{cases} 1 \text{ if $Q$ is good},\\
0 \text{ otherwise.}
\end{cases}
\]

Appropriate choices of $A, B> 0$ in \eqref{E:r1} and \eqref{E:r3} will be made below.  We make appropriate choices below.  The factor $|\log \epsilon|^{30}$ is so that boxes which violate the bounds are exceedingly rare if $\epsilon$ is small.  We have not attempted to optimize this part of the proof.  All of these requirements must be modified in two dimensions.  Also, the bound \eqref{E:r2} may be, in principle, tightened as the dimension increases.

Let $Y$ be an $L_0$-measurable connected and bounded set.  Suppose that $F, G$ are real valued functions on the collection of cubes in $Y$ which are $L_0$-measurable and define 
\[
[F;G]_{Y}=\sum_{\stackrel{Q \subset Y}{Q\in \QQ_{L_0}}} F(Q)G(Q).
\]

We will say that $Y$ is \textit{good} if 
\begin{equation}
\label{E:Dense}
[1-\Xi_{L_0}, 1]_Y = \sum_{\stackrel{Q \subset Y}{Q\in \QQ_{L_0}}} 1-\Xi_{L_0}(Q) \leq |\log \epsilon|^{-55} N^{L_0}_Y
\end{equation}
where $N^{L_0}_Y= \frac{|Y|}{|Q_{L_0}|}$.

For notational convenience we introduce the functions on $L_0$ measurable blocks
\begin{align*}
&F_{\lambda}(Q_{L_0})=\max_{\eta \in \NN} \frac{\|\gll_{Q_{\eta}}\|_2^2}{|Q_{\eta}|}, \quad \quad \quad F^{\nabla}_{\lambda}(Q_{L_0})=\max_{\eta \in \NN} \frac{\|\nabla\gll_{Q_{\eta}}\|_2^2}{|Q_{\eta}|},\\
&R(Q_{L_0})= \max_{\eta} \|\alpha\|_{\infty, Q_{\eta}}.
\end{align*}
An $L_0$-measurable connected and bounded set $Y$ will be called \textit{regular} if the following estimates hold for $\lambda\in \{0, L_0^{-2} \log L_0^8\}$:
\begin{align}
&\label{E:rr1} \left[F_\lambda^{\nabla};  \mathbf 1_{\{F_\lambda^{\nabla} \geq  \epsilon^2 |\log \epsilon|\}}\right]_{Y} 
\leq
 \epsilon^{\frac{9}4}  N_Y^{L_0}  \\
&\label{E:rr2}\left[F_\lambda;  \mathbf 1_{\{F_\lambda\geq  \epsilon^2[\lambda^{-
\frac 12}\wedge L_0]  |\log \epsilon|\}}\right]_{Y} 
\leq  |\log \epsilon|^2 N_Y^{L_0} \\
&\label{E:rr4}[R^2; 1_{R> |\log \epsilon|^{50}}]_{Y}  \leq |\log \epsilon|^{-75} N_Y^{L_0} \\
&\label{E:rr5} \sum_{\stackrel{Q_{L_0} \subset {Y} }{L_0-\text{measurable}}}\sum_{\eta} |\alpha(Q_{\eta})|\leq 
 L_0^{-\frac 32} \log^{50} L_0N_Y^{L_0}
\end{align}

\begin{definition}[Clean Regions]
An $L$-measurable connected and bounded set $Y$ will be called \textit{clean} if it is good and if $\delta(Y)$ is regular at the scales $L_0\in \{\ffrac \ell2, \ell, L\}$.  Otherwise, call the region dirty.
\end{definition}
\begin{remark}
Not all conditions are used at all scales.  Its just faster to formulate our requirements in a uniform way.
\end{remark}
\begin{remark}[Two Dimensions]
These bounds are not sufficient if $d=2$.  In this case, we need to control $\|\gll\|_{\infty}, \|\nabla \gll\|_{\infty}$ more carefully and also add requirements regarding the size of 
\[
\sum _{Q \subset Y} F_\lambda^{\nabla}(Q), \quad \quad \quad \sum _{Q \subset Y} F_\lambda(Q).
\] 
\end{remark}

The following proposition is fairly straightforward.  The parts of it we prove are postponed until \cref{S:Randomness}.
\begin{proposition}
\label{P:Dirt4}
Let $Y$ be connected $L$-measurable and bounded.   There exist $A, B>$ so that the following holds.  We can find $\epsilon_0 \in (0, 1)$ and a constant $c>0$ so that if $\epsilon \in (0, \epsilon_0)$
\begin{equation}
\mathbb P\left(Y \text{ is dirty}\right) \ls  e^{-  c |\log \epsilon|^2 N^L_Y}.
\end{equation}
\end{proposition}

Let
\begin{align*}
\AA&=\{Y \subset \Z^d: Y \text{ is connected $L$-measurable and is not clean}\},\\
\mathbb D&= \cup_{Y \in \AA} c(Y), \quad \quad \quad \text{and } \mathbb D_{\Lambda}= \mathbb D \cap \Lambda,
\end{align*}
where $\Lambda$ is a connected, bounded $L$-measurable subset of $\Z^d$ and 

For any $x \in \Z^3$, let $Q(x)$ denote the $L$-measurable box containing $x$ and define
\[
A(x)=\left\{\omega:  Q(x) \text{ is in $c(Y)$ for some $Y$ which is not clean}\right\}.
\]
We relegate the proof to \cref{S:Randomness}.
\begin{lemma}
\label{L:BadSetBound}
Choose $A, B$ so that the estimates in  \cref{P:Dirt4} hold.  There exists $\epsilon_0>0$ and $c>0$ so that if $\epsilon \in (0, \epsilon_0)$
\begin{equation}
\text{Var}(\mathbb D_\Lambda) \ls e^{-c\log^2 \epsilon}|\Lambda|
\end{equation}
Further, for any sequence $\Lambda_N \uparrow \Z^d$, $\omega$ a.s. there exists $N_0(\omega)>0$ so that
\[
 \frac{|\mathbb D_{\Lambda_N}|}{|\Lambda_N|}\ls e^{-c|\log \epsilon|^2} 
\] 
for all $N \geq N_0(\omega)$.
\end{lemma}

\section{Statement of the Peierls Estimate and Proof of \cref{T:Main2}}
\label{S:4}
Given a contour $\Gamma$, $\Gamma^*$ will denote the $(\Omega, \BB,
\mathbb P)$ event
\begin{equation} \label{E:1}
\Gamma^*=\{\omega: \delta(\Gamma) \text{ is clean and $c(\Gamma)$ is not
strictly contained in $\D$}\}
 \end{equation}

 If $\Gamma^*$ occurs, $\Gamma$ will be called a $*$-clean contour for short.
Given a spin configuration $(\sigma_{\Lambda_N},
\sigma_{\Lambda_N^c})$, let \[ \mathbb X(\Gamma_1^*, \dotsc,
\Gamma_m^*, \Gamma_{m+1},\dotsc, \Gamma_{m+n})= \cap_{i}  \mathbb
X(\Gamma_i) \] where $ (\Gamma_1, \dotsc, \Gamma_m)$ satisfy the
event defined in \eqref{E:1} and $( \Gamma_{m+1},\dotsc,
\Gamma_{m+n})$ do not.  Otherwise we define the right hand side to be
the empty set. As a variation of standard definitions, let us say
that \[ (\Gamma_1^*, \dotsc, \Gamma_m^*, \Gamma_{m+1},\dotsc,
\Gamma_{m+n}, \omega) \] are $*$-compatible if \[ \mathbb
X(\Gamma_1^*, \dotsc, \Gamma_m^*, \Gamma_{m+1},\dotsc, \Gamma_{m+n})
\neq \varnothing. \]

\begin{lemma}[The Peierls Estimate] \label{L:Contours}  There exist $\delta \in (0, 1)$ so that if $0< \xi<
\delta$, we can find $\epsilon_0=\epsilon_0(\xi)$ so that the following holds.  If $\epsilon< \epsilon_0$ there exists
$\beta_{\epsilon}$ so that for $\beta> \beta_{\epsilon}$ we have the estimate:

\noindent Let $N$ be fixed and consider the event $\mathbb
X(\Gamma_1^*, \dotsc, \Gamma_m^*, \Gamma_{m+1},\dotsc, \Gamma_{m+n})$.  Then
\[
\mu_{\Lambda_N}^{e_1}(\mathbb X(\Gamma_1^*, \dotsc, \Gamma_m^*,
\Gamma_{m+1},\dotsc, \Gamma_{m+n})) \leq e^{-\beta q \sum_{i=1}^m
|\Gamma_i|}
\]
where
\begin{equation}
\label{E:ContBon} q = C \xi^2 \epsilon^{2} |\log\epsilon|^{- 24}
 \end{equation}
 for some constant $C>0$.
\end{lemma}

This lemma is proved in several steps below: \cref{S:Energetics,,S:Groundstates,,S:Collar,,S:Peierls}.  \cref{T:Main2} is then completed via the
following Peierls contour counting argument.

\begin{proof}[Proof of \cref{T:Main2}]
The proof of \cref{T:Main2} is an application of \cref{L:BadSetBound,,L:Contours}.

Let $A, B>0$ be as in \cref{L:BadSetBound}.  Let $x$ is taken so that $Q_{L}(x) \cap \D = \varnothing$.  Note that this does not depend on $\Lambda_N$.  If $\Lambda_N \uparrow \Z^d$ in the Van Hove sense then eventually $Q_L(x) \subset \Lambda_N$.
Consider the event
\[
\{\sigma: \Psi_{Q_{L}(x)}(\sigma) \neq 1\} \subset \mathcal S_{\Lambda_N}.
\]
By definition of $\psi,
\Psi$ and the $e_1$ boundary condition, there exists a largest
contour $\Gamma$ so that $Q_{L}(x) \subset c(\Gamma)$. Moreover,
since $Q_{L}(x) \cap \D = \varnothing$, $\Gamma$ must be
clean.
Decomposing $\{\Psi_{Q_{L}(x)} \neq 1\}$
into disjoint subsets according to this largest contour we have:
\[
\mu_N^{e_1, \omega}(\Psi_{Q_{L}(x)} \neq 1) \leq \sum_{\Gamma,
\omega\: * \text{-compatible}: \: Q_{L}(x) \subset c(\Gamma)}
\mu_N^{e_1, \omega}(\mathbb X(\Gamma^*)).
\]
By \cref{L:Contours}, \[ \mu_N^{e_1, \omega}(\Psi_{Q_{L}(x)} \neq
1) \leq C \sum_{r \geq 1} r^{d/(d-1)} (2a_0)^{r} e^{-cq r} \leq
C_1e^{-c\frac{q}{2}} \] as long as $q> 2\log(2a_0)$.  The theorem now
follows easily.
\end{proof}

\section{Approximate Ground States in the Bulk of a Clean Contour}
\label{S:Groundstates}
Let $\Gamma$ be fixed and assume $\sigma \in \X(\Gamma)$.  
Let
\[
\begin{array}{lr}
\bar{\delta}(\Gamma)=\delta_{\ffrac L2}(\Gamma) \cap \Lambda_N, & \RR_{\ffrac \ell2 }(\Gamma)=\{Q\in \QQ_{\ffrac \ell2}: Q \cap \bar{\delta}(\Gamma) \neq \varnothing\}.
\end{array}
\]
Given a region $R \subset \Lambda_N$ let $ext$ denote the boundary condition which is set to $e_1$ on $\pO R \cap \Lambda_N^c$ and is free otherwise. 
With this definition set
\[
\textrm{E}_0(Q)=\max_{\sigma\in \mathcal S_Q} -\HH_Q(\sigma|\textrm{ext})
\]
The goal of the present section is the following Lemma:

\begin{lemma}[Bulk Ground States]
\label{L:GSd3}
Suppose that $\Gamma$ is a clean contour.
There exists $\epsilon_0>0$ and a constant $C>0$ so that if $0<\epsilon< \epsilon_0$ the following holds.  On $\bar{\delta}(\Gamma)$, there is a spin configuration $\sigma^{\bar{\delta}(\Gamma)}$ which satisfies
\[
 0 \leq \sum_{Q\in \RR_{\ffrac \ell2 }(\Gamma)}  \textrm{E}_0(Q) + \HH_{\bar{\delta}(\Gamma)}(\sigma^{\bar{\delta}(\Gamma)})  \ls \epsilon^2 |\log \epsilon|^{- 40}|\Gamma| 
\]
\noindent
Further,
\[
1- \sigma^{\bar{\delta}(\Gamma)}_y\cdot e_1\ls \sqrt{ \epsilon} \log ^{30} \epsilon 
\]
 for all $y \in \bar{\delta}(\Gamma)$
and
\[
\|\sigma^{\bar{\delta}(\Gamma)} \cdot e_2\|_{\infty, \pI \bar{\delta}(\Gamma)} \ls 
\epsilon |\log \epsilon|^{35} .
\]
\end{lemma}

In the remainder of this section we give the basic construction.  In \cref{SS:6} we estimate the errors made in the construction.  The construction proceeds by stitching together approximate ground states on the microscopic blocks $Q \in \RR_{\ffrac \ell2}(\Gamma)$ 
There are two cases to consider.

\noindent \textbf{Case 1:} $\Xi_{\ffrac \ell2}(Q)=0$.  We then set $\sigma^{\bar{\delta}(\Gamma)}_y \equiv e_1$ for all $y \in Q$.  This is unlikely to be (even nearly) optimal in $Q$ but because we are working in a regular contour  this error will not have a large contribution when calculating the energy difference between the configuration we are constructing now and any $\sigma \in \X(\Gamma)$ (this follows from a combination of \cref{L:FBE,,E:rr1,,E:rr5}).

\vspace{10pt}
\noindent \textbf{Case 2:} $\Xi_{\ffrac \ell2}(Q)=1$.  Provisionally we define, for all $y \in Q$, the angle
\[
\theta_y=
g^N_{Q, y}
\]
and the spin configuration
\begin{equation}
\label{E:GSS}
\sigma_{Q,y}= \left(\cos(\theta_y), \sin(\theta_y)\right).
\end{equation}
In order to stitch configurations in neighboring boxes we interpolate the function $\theta_y$ in $Q$ so that the angle is small near the boundary of a box.  
Let
\[
\tau_{Q, x}=
\frac{\dist(x, \pO Q)}{\sqrt{\ell}} \wedge 1
\]
 \text{ if $x \in Q$}
and
\[
\sigma^{\bar{\delta}(\Gamma)}_y= \left(\cos(\tau_y \: \theta_y), \sin(\tau_y \: \theta_y)\right) \text{ if $y \in Q$}.
\]
The reason for  this slightly complicated definition is that it guarantees $|\tau_y \:\theta_y| \leq 2 \ell^{-1}$ for all $y \in \pI Q$.

\section{Surgery at the Boundary of a Contour}
\label{S:Collar}
Suppose a contour $\Gamma$ is given and that $\sigma \in \X(\Gamma)$.  These will be fixed throughout this section.
Let 
\begin{align*}
&\mathfrak C(\Gamma)= \delta(\Gamma) \cap \textrm{sp}(\Gamma)^c \cap \Lambda_N, \quad  \quad \mathfrak C^{\pm}(\Gamma)= \{z \in \mathfrak C(\Gamma): \Psi_z(\sigma)= \pm 1\}\\
\end{align*}
i.e. $\mathfrak C(\Gamma)$ is a thickened version of the boundary of $\Gamma$.
Consider also the `middle' portion of $\mathfrak{C}(\Gamma)$:
\[
\begin{array}{ll}
\mathfrak{M}(\Gamma)= \{z \in \mathfrak C(\Gamma): \dist(z, [\pO \mathfrak C(\Gamma) \backslash \Lambda_N^c]) \geq \frac L2- 100\},& \MM(\Gamma)=\{Q\in \QQ_{\ffrac \ell2}: Q \cap \textrm{M}(\Gamma) \neq \varnothing\}, \\
\textrm{M}(\Gamma)= \cup_{Q \in \MM(\Gamma)} Q, & \textrm{M}^{\pm}(\Gamma)= \textrm{M}(\Gamma)\cap \mathfrak C^{\pm}(\Gamma).
\end{array}  
\]
We will often drop reference to $\Gamma$ in these sets below as $\Gamma$ is fixed throughout.

The main result of this section is the following.
\begin{lemma}
\label{L:Collar}
There exists $\epsilon_0>0$ so that the following holds for all $\epsilon \in (0, \epsilon_0)$.  Let $\Gamma$ be a clean contour.

Let $\sigma\in \X(\Gamma)$. From $\sigma$, we can construct a spin configuration $\sigma^{\mathfrak C}$ with the following properties.
\begin{enumerate}
\item $\sigma_x = \sigma^{\mathfrak C}_x$ if $\dist(x, \mathfrak C(\Gamma)) \geq \frac {3L}{2}$.

\item If $x \in \pO \textrm{M}^{\pm}(\Gamma)$ then $\|\sigma^{\mathfrak C}_x \mp e_1\|_2\leq \epsilon |\log \epsilon|$.

\item For any choice of boundary condition ${e_1}$ compatible with $\Gamma$ and for which $\Sp(\Gamma) \subset \Lambda_N$,
\begin{equation}
-\HH_{\Lambda_N}(\sigma^{\mathfrak C}|{e_1}) \geq -\HH_{\Lambda_N}(\sigma|{e_1})  - C\epsilon^2 |\log \epsilon|^{- 25}|\Gamma|.
\end{equation}
\end{enumerate}
\end{lemma}
The proof of this lemma  consists of four separate modifications of $\sigma$.  We will first present the essential steps and Lemmas used in the construction.  Proofs will be given in later sections.

\subsection{Modification 1}
Let $\mathfrak A^{\pm}$ denote the smallest set containing $\mathfrak C^{\pm}$ so that
\[
\sgn(\sigma_x \cdot e_1)= \pm 1 \text{ for all $x \in \pO \mathfrak A^{\pm}$}.
\]
The following is a consequence of \cref{L:Defect}.  Recall that $\xi$ defines the cutoff for whether block averages of spins are good or bad.
\begin{proposition}
\label{C:LocRef}
There exist $\xi_0> 0$ and $\epsilon_0(\xi_0)$ so that if $\epsilon \in (0, \epsilon_0)$ and $\xi \in (0,  \xi_0)$,
\[
\mathfrak A^{\pm}\cup \pO \mathfrak A^{\pm} \subset \{x: \dist(x, \mathfrak C^{\pm}\leq L\}
\]
\end{proposition}
On the basis of this fact it is convenient to make a first modification of $\sigma$, into $\sigma^1_x$, as follows:  If $x \in \mathfrak A^{\pm}$ and $\sgn(\sigma_x\cdot e_1)= \mp1$,  set $\sigma^1_x$ to be the  reflection of $\sigma_x$ across the $e_2$ axis.  Otherwise set $\sigma^1_x=\sigma_x$.
Let us record, without proof, some simple properties of this transformed configuration.
\begin{lemma}
\label{L:>0}
Given $\sigma \in \X(\Gamma)$, let $\sigma^1$ be given as above.
\begin{enumerate}
\item
$\sigma^1\in \X(\Gamma)$.

\item
If $x \in \mathfrak C^{\pm}$, $\sgn(\sigma^1_x\cdot e_1) = \pm 1$ if $\sigma_x \cdot e_1 \neq 0$.

\item
\[
|\sigma^1(Q_\ell) \cdot e_1|\geq |\sigma(Q_\ell) \cdot e_1|
\]
whenever $Q_{\ell} \subset \mathfrak C$.

\item
If $\Gamma$ is compatible with the boundary condition ${e_1}$,
\[
-\HH_{\Lambda_N}(\sigma^1|{e_1}) \geq -\HH_{\Lambda_N}(\sigma|{e_1}).
\]
For any region $R \subset \Lambda_N$
\[
\EE_R(\sigma^1) \leq \EE_R(\sigma).
\]

\item
The number of spin configurations $\sigma\in \X(\Gamma)$ which map to the same $\sigma^1$ is bounded by $c_d^{|\Gamma|}$ for some universal constant $c_d>0$.
\end{enumerate}
\end{lemma}

\subsection{Modification 2}
From now on, we will work with $\sigma^1$.  Ideally, we want to use \cref{L:BLayer} to show that by replacing $\sigma^1$ by the optimizer of  $-\HH_{\mathfrak C}(\sigma|\sigma^1)$ in $\mathfrak C$ the resulting configuration behaves as claimed in \cref{L:Collar}.  However, because the random field may behave rather poorly in ( relatively small) subsets of $\mathfrak C$ this cannot be done directly.  

Set 
\begin{align}
&D= \cup_{Q\in \QQ_L: Q\subset \mathfrak C \text{ and }  \Xi_{L}(Q)=0} Q\\
&\mathscr D= \{x: \dist(x, D) \leq 5L\}, \quad \quad \mathscr D^{\pm}= \mathscr D \cap \mathfrak C^{\pm}(\Gamma), \\
&\mathfrak D_{L'}^{\pm}=\{x \in \mathscr D^{\pm}: \dist(x, \pO \mathscr D^{\pm}) \geq L' \}.
\end{align}
In the following, we work with the regions $\mathfrak D_{\ffrac L{12}}^{\pm} \subset \mathfrak D_{\ffrac L{16}}^{\pm}$.

To prepare the way for \cref{L:BLayer}, we first we will brutally change $\sigma^1$ in $\mathscr D^{\pm}$.
Define
\[
\tau_x=
\begin{cases}
\frac{16 \dist(x, \mathfrak D^{\pm}_{\ffrac L{16}})}{L} \wedge 1 \quad \text{ for $x \in \mathscr D^{\pm} \backslash \mathfrak D_{\ffrac{L}{16}}^{\pm}$},\\
0 \quad \text{ for $x \in \mathfrak D_{\ffrac{L}{16}}^{\pm}$},\\
1 \quad \text{ otherwise.}
\end{cases}
\]
Representing $\sigma_x= (\cos(\theta_x), \sin(\theta_x))$ with $\theta_x \in [- \frac {\pi}{2}, \frac{ \pi}{2}]$ (resp. $\theta_x \in [\frac {\pi}{2}, \frac{3\pi}{2}]$) if $x\in \mathfrak C^{+}$ (resp.  if $x\in \mathfrak C^{-}(\Gamma)$) let
\[
\sigma^{2}_x=
\begin{cases}
\sigma_x^1 \quad \text{ if $x \notin \mathscr D^{\pm}$},\\
(\cos(\tau_x\theta_x), \sin(\tau_x \theta_x)) \quad \text{ if $x \in \mathscr D^+$},\\
(\cos(\tau_x[\theta_x-\pi]+ \pi), \sin(\tau_x [\theta_x- \pi] + \pi)) \quad \text{ if $x \in \mathscr D^-$}.
\end{cases}
\]

\begin{lemma}
\label{L:INT2}
There is $\epsilon_0$ so that if $\epsilon \in (0, \epsilon_0)$ and if $\Gamma$ is a clean contour,  
\[
\left|-\HH_{\Lambda_N}(\sigma^{2}|{e_1})+ \HH_{\Lambda_N}(\sigma^1|{e_1})\right|
\ls \epsilon^2|\log \epsilon|^{-25} |\Gamma|
\]
\end{lemma}

\subsection{Modification 3}
Set
\[
\mathscr C^{\pm}= \mathfrak C^{\pm} \backslash \mathfrak D^{\pm}_{\ffrac L{12}}.
\]
Our next goal is to modify $\sigma^{2}$ in $\mathscr C^{\pm}$ in such a way that the resulting configuration has small projection onto the vertical ($e_2$) axis for all $x \in \pO \textrm{M}$ without too much cost in energy.  Note this aim was achieved already within
$ \mathfrak D_{\ffrac{L}{16}}^{\pm}$.

To be concrete with our calculations, we restrict attention to $\mathscr C^+, \mathfrak D^{+}_{\ffrac L{12}}$; the region $\mathscr C^-, \mathfrak D^{-}_{\ffrac L{12}}$ is treated similarly.
Let $\gld_{x, \mathscr C^+}=[-\Delta+ \lambda]^{-1} \alpha$ with dirichlet boundary conditions on $\pO\mathscr C^+$ and with \begin{equation}
\label{E:lambda3}
\lambda=L^{-2} \log^8 L.
\end{equation}

Because of Modification 1, for $x \in \mathscr C^+ \cup \pO \mathscr C^+$, we can express $\sigma^2_x$ in angular variables as $\sigma^2_x=(\cos(\theta_x),\sin(\theta_x))$ with $\theta_x \in [-\frac{\pi}{2}, \frac{\pi}{2}]$.
Consider change of variables
\[
\phi_x= \theta_x - \cos(\theta_x)\gld_{x, \mathscr C^+} \text{ for $x \in \mathscr C^+$}.
\]
We show later that $\|\gld_{\mathscr C^+}\|_{\infty}\ll \ffrac\pi{12}$( using \cref{L:TechBadRandomeness} and the definition of $\mathscr C^+$).  Thus the transformation is nonsingular and preserves the half-space $\sigma\cdot e_1 \geq 0$ in spin space.  

Let
\[
\mathscr F(\Gamma)=\left \{ x\in  \mathfrak C^+: \dist(x, \pO \mathfrak C^+\cap \Lambda_N) \geq \ffrac{L}{8}\right \}
\]
and let $\mathfrak f$ be the smallest subset of $\Z^3$ containing $\mathscr F$ and so that
\[
|\phi_x | \leq \ffrac{\pi}{6}
\]
for all $x \in \pO \mathfrak f$.  

For $\epsilon$ small enough, our definition of $\mathfrak C$, \cref{L:Defect} and the fact that $\|\gld_{\mathscr C^+}\|_{\infty}\ll \ffrac\pi{12}$ together imply the following.  This is the second point at which we use low dimensionality of the lattice $\Z^3$.
\begin{proposition}
There exist $\xi_0> 0$ and $\epsilon_0(\xi_0)$ so that if $\epsilon \in (0, \epsilon_0)$ and $\xi \in (0, \xi_0)$,
\label{C:E1}
\[
\mathfrak f \subset \left \{x \in \mathfrak C^+: \dist(x,  \pO \mathfrak C \cap \Lambda_N) \geq \ffrac L9 \right\} .
\]
\end{proposition}

Set $\mathfrak g^+:= \mathfrak f \cap \mathscr C^+$. After the change of variables in $\mathscr C^+$, according to \cref{L:BLayer} it is sufficient to optimize (after restricting $\KK_{\mathscr C^+}$ to $\mathfrak g^+$)
\begin{equation}
\label{E:AuxHam}
\KK_{\mathfrak g^+}(\varphi|\phi)=\sum_{e \cap \mathfrak g^+ \neq 0} \cos(\nabla_e\varphi)-1 + \frac 14 \sum_{x\in  \mathfrak g^+} m_x \cos^2(\varphi_x)
\end{equation}
where $m_x= \sum_{y \sim x} [\gld_{\mathscr C^+, y}-\gld_{\mathscr C^, x}]^2$.
In \cref{E:AuxHam}, we neglected the boundary terms in \cref{E:BL} and also the $Error^{\lambda, D}_{\mathscr C^+}$ term.  The former plays no role in the optimization and the latter will be bounded carefully in the proof of \cref{L:INT3}.
The advantage of \cref{E:AuxHam} over the original Hamiltonian is that  stationary points of $\KK_{\mathfrak g^+}(\cdot| \phi)$ behave as solutions to discrete elliptic PDE's with (random) mass.

Let $\vartheta$ be the maximizer of $\KK_{\mathfrak g^+}(\varphi|\phi)$ on $\mathfrak g^+$ and let
\begin{equation}
\phi'_x=\begin{cases}
\vartheta_x \text{ if $x \in \mathfrak g^+$}\\
\phi_x \text{ if $x \in \mathfrak C^+ \backslash \mathfrak g^+$}.
 \end{cases}
 \end{equation}

Inverting the change of variables, define  $\Phi^+_x$ by
\[
\phi'_x=\Phi^+_x- \cos(\Phi^+_x) \gld_{\mathscr C^+, x}
\]
We construct the configuration $\Phi^-$ close to $\pi$ in a similar way on $\mathfrak g^-$ (fluctuating around the angle $\pi$).
Modification 3 of the original spin configuration is defined by
\[
\sigma^{3}_x:=
\begin{cases}
\left(\cos(\Phi^+_x), \sin(\Phi^+_x)\right) \text{ for $x \in \mathscr C^+$},\\
\left(\cos(\Phi^-_x), \sin(\Phi^-_x)\right) \text{ for $x \in \mathscr C^-$},\\
\sigma^{2}_x \text{ otherwise.}
\end{cases}
\]
\begin{lemma}
\label{L:INT3}
There exists $\epsilon_0$ so that if $\epsilon< \epsilon_0$
\begin{equation}
\label{E:Trans1}
-\HH_{\Lambda_N}(\sigma^{3}|{e_1})\geq -\HH_{\Lambda_N}(\sigma|{e_1})
- C \epsilon^2 |\log \epsilon|^{-25}|\Gamma|.
\end{equation}
Moreover, 
\begin{align*}
&\sigma^{3} \equiv \sigma^2 \text{ for $x \in \Lambda_N\backslash \mathfrak g^+\cup \mathfrak g^-$},\\
&|\Phi^+_x- \gld_{\mathscr C^+, x}| \ls \exp(-c \log^4 \epsilon) \text{ if $\dist(x, \textrm{M}^+)< \ffrac L5$},
\end{align*}
and similarly for $\Phi^-_x$.
\end{lemma}

\subsection{Modification 4}
The configuration $\sigma^{3}$ is almost what we want. We make a final modification into $\sigma^{\mathfrak C}$ so that the resulting configuration is close to $\pm e_1$ on the appropriate components of $\pO \textrm{M}$.

By construction
\[
\mathfrak N:= \left\{x \in \mathfrak C(\Gamma): \dist(x, \pO \mathfrak C\cap \Lambda_N)\geq \frac {4L}{5}\right\} \subset \mathfrak D_{\frac L{12}}^{+} \cup \mathfrak g^+ \cup  \mathfrak D_{\frac L{12}}^{-} \cup \mathfrak g^-
\]
and, obviously, $\textrm{M}$ is contained in $\mathfrak N$.
Let $(\mathfrak c_i)_{i \in [I]}$ denote  the connected components of $\mathfrak N$ and let $\textrm{M}_i = \mathfrak c_i \cap \textrm{M}(\Gamma)$.

To achieve the stated goal we interpolate in a way similar to that employed in \cref{S:Groundstates}.  For $x \in \mathscr C^{\pm}$ let 
\[
\tau_x=
\begin{cases}
\frac{\dist(x,\: \pO \textrm{M}_i)}{\sqrt{\ell}} \wedge 1 \quad \text{ if $x \in \mathfrak c_i$},\\
1 \quad \text{ otherwise.}
\end{cases}
\]

Let the angle $\theta_x$ be defined by  $\sigma^3_x= (\cos(\theta_x), \sin(\theta_x))$.
 If $\mathfrak c_i \subset \mathfrak C^{+}(\Gamma)$ we may take $\theta_x \in [-\frac \pi2, \frac \pi2]$ and then set
\[
\sigma_{i,x} =(\cos(\tau_x \theta_x), \sin(\tau_x \theta_x)).
\]
If $\mathfrak c_i \subset \mathfrak C^{-}(\Gamma)$ we may take for $\theta_{x} \in [\frac \pi2, \frac{3\pi}{2}]$ and then
define
\[
\sigma_{i, x}=(\cos(\tau_x [\theta_x-\pi]+ \pi), \sin(\tau_x[\theta_x- \pi]+\pi ).
\]
Finally we set
\[
\sigma^{\mathfrak C}_x=
\begin{cases}
\sigma_{i, x} \text{ if $x \in \mathfrak c_i \cap \mathfrak G$},\\
\sigma^{3}_x \text{ otherwise},
\end{cases}
\]

\begin{lemma}
\label{L:INT4}
There exists $\epsilon_0$ so that if $\epsilon< \epsilon_0$ and if $\Gamma$ is a clean contour.
\begin{equation}
\label{E:D3Summary}
\left|\HH_{\Lambda_N}(\sigma^{3}|{e_1} )-\HH_{\Lambda_N}(\sigma^{\mathfrak C}|{e_1})\right| \\
\ls \epsilon^{2} |\log \epsilon|^{-40} |\Gamma|.
\end{equation}
Moreover, if $x \in \pO \textrm{M}^{\pm}(\Gamma)$ then $\|\sigma^{\mathfrak C}_x \mp e_1\|_2\leq \epsilon |\log \epsilon|$.
\end{lemma}

Let us finish this section with the observation:
\begin{proof}[Proof of \cref{L:Collar}]
The statement of \cref{L:Collar} follows from the construction, collecting the estimates from \cref{E:Trans1} and \cref{E:D3Summary}. 
\end{proof}

\section{Gluing Configurations and the Proof of the Peierls Estimate}
\label{S:Peierls}
Given $\sigma \in \X(\Gamma)$,
the next step in our construction is to attach $\sigma^{\bar{\delta}(\Gamma)}$ as constructed in \cref{S:Groundstates} to a modification of the configuration $\sigma^{\mathfrak C(\Gamma)}$ as constructed from $\sigma$ in \cref{L:Collar}.

We shall say that $\Gamma$ is a
$\pm$ contour if $\Psi_z(\Gamma)=\pm 1$ on $\delta_{ext}(\Gamma)$ (recall that ${\delta_{ext}(\Gamma)}$ is the exterior boundary component of $\Gamma$ as defined below \cref{D:Contour}).  To obtain the desired configuration on $c(\Gamma)= \delta(\Gamma)\cup \Int(\Gamma)$, we distinguish whether $\Gamma$ is a $\pm$ contour.  

Restricting the spin configuration  $\sigma^{\mathfrak C(\Gamma)}$ to $\Lambda_N \backslash \bar{\delta}(\Gamma)$,   we produce a new spin configuration $\sigma^{*}$ on $\Lambda_N \backslash \bar{\delta}(\Gamma)$ as follows.  If $\Gamma$ is a $+$ contour, then on each interior component with $\Psi^{(0)}|_{\delta^i_{in}(\Gamma)}=-1$ we reflect all spins across the $e_2$-axis.  For each of the remaining components of $\Lambda_N\backslash \bar{\delta}(\Gamma)$ we keep the spin configuration fixed.  

We define
\[
\textrm{S}^{+}_{\Gamma, y}:=\textrm{S}^{+}_{\Gamma, y}(\sigma)=
\begin{cases}
\sigma^{*}_y \text{ if $y \in  \bar{\delta}(\Gamma)^c$}\\
\sigma^{\bar{\delta}(\Gamma)}_y \text{if $y \in \bar{\delta}(\Gamma)$}.
\end{cases}
\]
If instead $\Gamma$ is a $-$ contour, we reflect spins on interior components with $\Psi^{(0)}=1$ and use the reflection of $\sigma^{\bar{\delta}(\Gamma)}_y$ across the $e_2$ axis in place of  $\sigma^{\bar{\delta}(\Gamma)}_y$, calling the resulting configuration $\textrm{S}^{-}_{\Gamma}$.
The proof of the next lemma appears in \cref{S:PP}.
\begin{lemma}
\label{L:EComp}
There exists $\epsilon_0\in (0,1)$ so that for all $\epsilon \in (0, \epsilon_0)$ the following holds.  Suppose that $\Gamma$ is a clean contour and that $\sigma \in \X(\Gamma)$.
Then
\begin{equation}
-\HH_{\Lambda_N}(\textrm{S}^{\pm}_{\Gamma}(\sigma)|{e_1})+ \HH_{\Lambda_N}(\sigma|{e_1})\gtrsim \xi^2 \epsilon^{2} |\log \epsilon|^{-24}|\Gamma|
\end{equation}
\end{lemma}

\subsection{The Peierls Estimate}
We are now ready to derive a Peierls estimate for this system.
By definition, if $\sigma \in \mathbb X(\Gamma)$, the restriction of
$\Psi_z(\sigma)$ to each of the components $\delta_{ext}(\Gamma),
\delta^{i}_{in}(\Gamma)$ is constant. 
Let $ R^{\pm} (\Gamma)= R^{\pm}\backslash \delta(\Gamma)$.

In order to deal with contours adjacent to $\Lambda_N^c$ let $\delta_N(\Gamma)= \delta(\Gamma)\cap \Lambda_N$.  Given that $\Gamma$ is a clean $+$ contour, we say that a 
spin configuration $\sigma_{\delta_N(\Gamma)^c}\in \mathcal S_{\delta_N(\Gamma)^c}$ is compatible with
$\Gamma$ if $\sigma_{\delta_N(\Gamma)^c}|_{\Lambda_N^c} \equiv e_1$ and if 
\[
\mu_{\delta_N(\Gamma)}^{\sigma_{\delta_N(\Gamma)^c}}\left(\psi_z(\sigma'_{\delta_N(\Gamma)})=
\psi_{\Gamma}(z) \text{ for } z \in \Sp(\Gamma) \text{ and }
\Psi_z(\sigma'_{\delta_N(\Gamma)})= \pm 1\text{ for } z \in
R^{\pm}(\Gamma)\right)\neq 0
\] 

Note here that $
\Psi_z$ implicitly takes as an argument
the extended configuration $(\sigma'_{\delta_N(\Gamma)},
\sigma_{\delta_N(\Gamma)^c})$ although we will suppress
this detail below.

For any such $\sigma_{\delta_N(\Gamma)^c}$, let
\[ W(\Gamma;
\sigma_{\delta_N(\Gamma)^c}):=
 \frac{\mu_{\delta_N(\Gamma)}^{\sigma_{\delta_N(\Gamma)^c}}\left(\psi_z(\sigma'_{\delta_N(\Gamma)})=
 \psi_{\Gamma}(z) \text{ for } z \in \Sp(\Gamma) \text{ and }
 \Psi_z(\sigma'_{\delta_N(\Gamma)})= \pm 1\text{ for } z \in
 R^{\pm}(\Gamma)\right)}{\mu_{\delta_N(\Gamma)}^{ \mathscr R^+ \cdot \sigma^1_{\delta_N(\Gamma)^c}}\left( \Psi_z(\sigma'_{\delta_N(\Gamma)})=
 1\text{ for } z \in \delta_N(\Gamma)\right)}.
\]
Here
$\sigma^1_{\delta_N(\Gamma)^c}$ denotes the
boundary condition obtained from \cref{L:>0} and $\mathscr R^+ \cdot \sigma^1_{\delta_N(\Gamma)}$ is obtained by making a global reflection of $\sigma^1$ on each component of $R^- \cap [c(\Gamma)\backslash \delta(\Gamma)]$.
Let \[ \|W(\Gamma; \cdot)\| =
\sup_{\{\sigma_{\delta_N(\Gamma)^c} \text{ compatible}\}} W(\Gamma;
\sigma_{\delta_N(\Gamma)^c}). \] 
Notions for $-$ contours are
defined similarly with the provisos for $+$ and $-$ reversed.  

\begin{lemma}
\label{L:Key} There exist $\delta, \epsilon_0, \beta_0, C>0$ so that if
$\epsilon< \epsilon_0$, $\beta> \beta_{\epsilon}> \beta_0$, $0< \xi<
\delta$ and $\Gamma$ is a clean contour
\[ \|W(\Gamma;
\cdot)\| \leq e^{-q |\Gamma|} \] where
\[
q =
C \beta \xi^2 \epsilon^2 |\log \epsilon|^{-24}.
\]
 \end{lemma}
\begin{proof}
Assume for concreteness that $\Gamma$ is a $+$ contour.
Given a configuration $\sigma_{\delta_N(\Gamma)^c}$ compatible with $\Gamma$ let $\sigma'_{\delta_N(\Gamma)} \in \mathcal S_{\delta_N(\Gamma)}$ so that $\sigma':=(\sigma'_{\delta_N(\Gamma)}, \sigma_{\delta_N(\Gamma)^c}) \in \X(\Gamma)$.  Let $\textrm{S}^{+}_{\Gamma}(\sigma')$ denote the output of our construction from just above \cref{L:EComp}. Dropping the super/subscripts,
\[
\textrm{S}(\sigma')|_{\delta_N(\Gamma)^c} =  \mathscr R^+ \cdot \sigma^1_{\delta_N(\Gamma)^c}
\]
where $\sigma^1$ is the first modification of $\sigma'$ as defined in \cref{S:Collar}.  Note that 
$\textrm{S}(\sigma')|_{\delta_N(\Gamma)}$ is an element of the event in the denominator of  $W(\Gamma, \sigma_{\delta_N(\Gamma)^c})$.

Let
\[
\textrm{F}_{\sigma'}=\{\sigma \in \mathcal S_{\delta_{N}(\Gamma)}: \|\sigma_x-\textrm{S}(\sigma')_x\|_2 \leq \epsilon^3 \:\: \forall x \in \delta_N(\Gamma)\}.
\]
Then
\[
\mu_{\delta_N(\Gamma)}^{ \mathscr R^+ \cdot \sigma^1_{\delta_N(\Gamma)^c}}\left( \Psi_z(\sigma'_{\delta_N(\Gamma)})=1 \text{ for } z \in \delta_N(\Gamma) \right) \geq \mu_{\delta_N(\Gamma)}^{ \mathscr R^+ \cdot \sigma^1_{\delta_N(\Gamma)^c}}\left(\textrm{F}\right).
\]
Using \eqref{E:rr4}, we have
\[
\epsilon \sum_{x\in \delta(\Gamma)} |\alpha_x| \ls |\Gamma|,
\]
Using this bound, it follows from the definition of $F_{\sigma'}$ that
\[
|-\HH(\textrm{S}|_{\delta_N(\Gamma)}|  \mathscr R^+ \cdot \sigma^1_{\delta_N(\Gamma)^c})+ \HH(\sigma| \mathscr R^+ \cdot \sigma^1_{\delta_N(\Gamma)^c})| \ls \epsilon^{ 3}|\Gamma|.
\] 

By \cref{L:EComp}, we therefore have
\[
W(\Gamma, \sigma_{\delta(\Gamma)^c}) \leq e^{[-\beta f(\epsilon)+ g(\epsilon)]|\Gamma|}
\]
where
\[
f(\epsilon)= C_1\xi^2\epsilon^2|\log \epsilon|^{-24}
\]
accounts for internal energy difference between the numerator and denominator and
\[
g(\epsilon)=C_2|\log \epsilon|
\]
accounts for the entropy difference.
The lemma follows immediately.
\end{proof}

\begin{proof}[Proof of \cref{L:Contours}] Let $N$ be fixed and consider the event $\mathbb
X(\Gamma_1^*, \dotsc, \Gamma_m^*, \Gamma_{m+1},\dotsc, \Gamma_{m+n})$ where the contours are assumed compatible with the boundary condition ${e_1}$.  Then we claim \[
\mu_{\Lambda_N}^{\ra}(\mathbb X(\Gamma_1^*, \dotsc, \Gamma_m^*,
\Gamma_{m+1},\dotsc, \Gamma_{m+n})) \leq \prod_{i=1}^m c_d^{|\Gamma_{i}|} \|W(\Gamma_i;
\cdot)\| \] 
where $c_d$ is the constant from \cref{L:>0}.
Once this is justified, the Lemma is proved by
application of \cref{L:Key}.

The proof of this claim proceeds by induction on $m$.  Interpreting
an empty product as $1$, the case $m= 0$ there is nothing to prove,
so we proceed to the induction step. Suppose the claim is true for
any $*$-compatible system with $m=k$ clean contours and $n$ dirty
contours.  Given a set $\{\Gamma_1^*, \dotsc, \Gamma_{k+1}^*,
\Gamma_{k+2},\dotsc, \Gamma_{k+n+1})\}$ of $*$-compatible contours
and reordering as necessary, we may assume \[ c(\Gamma_{k+1}) \cap
\cup_{i=1}^k \Sp\{\Gamma_i\} = \varnothing. \]

Assume for concreteness that $\Gamma_{k+1}$ is a $+$ contour.  The
argument in the case of a $-$ contour proceeds in a similar manner.
Let $\{\tilde \Gamma_1, \dotsc, \tilde \Gamma_r\}$ denote the set of
contours among $\{\Gamma_{m+1},\dotsc, \Gamma_{m+n}\}$ with
$\delta(\tilde{\Gamma}_i) \subset \Int(\Gamma_{k+1})$ and $\{
\Gamma'_1, \dotsc, \Gamma'_{n-r}\}$ denote the rest.  Then we have
\[
\mathbb X(\Gamma_1^*, \dotsc, \Gamma_{k+1}^*, \Gamma_{k+2},\dotsc,
\Gamma_{k+n+1})= \mathbb X(\Gamma_1^*, \dotsc, \Gamma_k^*,
\Gamma'_{1},\dotsc, \Gamma'_{n-r}) \cap \mathbb X(\Gamma_{k+1}^*,
\tilde \Gamma_1, \dotsc, \tilde \Gamma_r) \]

Using the DLR equations
we have
\begin{multline} \mu_{\Lambda_N}^{\ra}(\mathbb X(\Gamma_1^*,
\dotsc, \Gamma_{k+1}^*, \Gamma_{k+2},\dotsc, \Gamma_{k+n+1}))\\= \left\langle \mathbf 1_{\mathbb X(\Gamma_1^*, \dotsc, \Gamma_k^*,
\Gamma'_{1},\dotsc, \Gamma'_{n-r})} 
\langle \mathbf 1_{\mathbb
X(\Gamma_{k+1}^*,\tilde{\Gamma}_{1},\dotsc,
\tilde{\Gamma}_{r})}\rangle_{c_N(\Gamma_{k+1})}^{\sigma_{c^c_N(\Gamma_{k+1})}} \right
\rangle_{{N}}^\ra.
\end{multline}
Note here that 
\[
\langle \mathbf 1_{\mathbb
X(\Gamma_{k+1}^*,\tilde{\Gamma}_{1},\dotsc,
\tilde{\Gamma}_{r})}\rangle_{c_N(\Gamma_{k+1})}^{\sigma_{c^c_N(\Gamma_{k+1})}} 
\]
may be $0$ depending on $\sigma_{c^c_N(\Gamma_{k+1})}$.

If $\Gamma= (\Sp(\Gamma), \psi_{\Gamma})$, let $-\Gamma=
(\Sp(\Gamma), -\psi_{\Gamma})$ and let \[
T_{\Gamma_{k+1}}(\tilde{\Gamma}_{\ell})= \begin{cases}
\tilde{\Gamma}_{\ell} \text{ if $\delta_{ext}(\tilde{\Gamma}_{\ell})
\subset \delta^+_{in}(\Gamma_{k+1}))$}\\ -\tilde{\Gamma}_{\ell}
\text{ otherwise}. \end{cases} \] 

Recall the map $\sigma\mapsto \sigma^1$ from  \cref{L:>0}.  It is straightforward to show that
\begin{multline}
\label{E:BoundGood}
\langle \mathbf 1_{\mathbb
X(\Gamma_{k+1}^*,\tilde{\Gamma}_{1},\dotsc,
\tilde{\Gamma}_{r})}\rangle_{c_N(\Gamma_{k+1})}^{\sigma_{c^c_N(\Gamma_{k+1})}}\\ \leq
\|W(\Gamma_{k+1}; \cdot)\|
\langle \mathbf 1_{\mathbb
X(T_{\Gamma_{k+1}}(\tilde{\Gamma}_{1}),\dotsc,
T_{\Gamma_{k+1}}(\tilde{\Gamma}_{r}))}\rangle_{c_N(\Gamma_{k+1})}^{\sigma^1_{c^c_N(\Gamma_{k+1})}}.
\end{multline}

Because the map is at most $c_d^{|\Gamma|}$-to-$1$  for some \textit{universal} constant $c_d>0$,  (see \cref{L:>0}), 
\begin{multline}
\left\langle \mathbf 1_{\mathbb
X(\Gamma_1^*, \dotsc, \Gamma_k^*, \Gamma'_{1},\dotsc, \Gamma'_{n-r})}
\langle \mathbf 1_{\mathbb
X(T_{\Gamma_{k+1}}(\tilde{\Gamma}_{1}),\dotsc,
T_{\Gamma_{k+1}}(\tilde{\Gamma}_{r}))}\rangle_{c_N(\Gamma_{k+1})}^{\sigma_{c^c_N(\Gamma_{k+1})}}
\right\rangle_{N}^\ra\\
 \leq c_d^{|\Gamma_{k+1}|}  \|W(\Gamma_{k+1}; \cdot)\| \mu_{\Lambda_N}^{\ra}(\mathbb
X(\Gamma_1^*, \dotsc,  \Gamma_{k}^*, \Gamma'_{1},\dotsc,
\Gamma'_{n-r},T_{\Gamma_{k+1}}(\tilde{\Gamma}_{1}),\dotsc,
T_{\Gamma_{k+1}}(\tilde{\Gamma}_{r})))
\end{multline}
and the induction
step is proved. \end{proof}

\section{Energetic Estimates}
This section is the engine room for the entire paper, see in particular \cref{L:BLayer}.
\label{S:Energetics}
\begin{proposition}
\label{L:FBE}
Let us suppose that $L' \in \N$ is fixed.  Then for any $\lambda\geq 0$ and any cube $Q$ of sidelength $L'$
\begin{multline}
\label{H:FML}
 -\HH_{Q}(\sigma)= \frac 12 \left[\EE_{Q}(\gln_{Q})-\EE_{Q}(\sigma \cdot e_1)- \EE_{Q}(\sigma\cdot e_2 -\gln_{Q})\right]\\
+\lambda \sum_{x \in Q} \gln_{Q, x} e_2 \cdot \sigma_x+O(\epsilon |\alpha(Q)||Q|).
 \end{multline}
 Similarly, for any finite region $R\subset \Z^d$ and any $\lambda>0$
 \begin{multline}
 -\HH_{R}(\sigma|\sigma_0) \label{H:DML}= -\frac 12\sum_{e \cap R\neq \varnothing} \left[\nabla_e\sigma \cdot e_1\right]^2+ \left[\nabla_e\sigma \cdot e_2 -\nabla_e\gld_{R})\right]^2 \\
 + \frac 12 \sum_{e \cap R\neq \varnothing}\left[\nabla_e\gld_{R})\right]^2+ \lambda \sum_{x \in R} \gld_{R, x} e_2 \cdot \sigma_x +\sum_{\stackrel{y \in \pO R}{x \sim y, x \in \pI R}} \gld_{x, R}e_2 \cdot \sigma_y.
\end{multline}
\end{proposition}

\begin{proof}
Consider the free boundary condition case $\lambda=0$.  Solving the equation $-\Delta^N g= \hat{\alpha}$ in $Q$, we can write
\[
-\HH_{Q}(\sigma) = -\frac 12 \EE_{Q}(\sigma)+ \sum_{e \subset Q} \nabla_e g^N_{Q} e_2  \cdot \nabla_e\sigma +O(\epsilon |\alpha(Q)||Q|).
\]
Completing the square gives the expression first expression.    The remaining cases are similar.
\end{proof}

For notational convenience, let  us write $\sigma_x= (\cos(\theta_x), \sin(\theta_x))$ with $\theta_x$ only defined modulo $2 \pi$.  For any pair of nearest neighbors $x, y$, let $[\nabla_e\theta)]^2_{2\pi}$ denote the minimum of $[\nabla_e\theta)]^2$ over all pairs of angles in the equivalence classes which give the vectors $(\cos(\theta_x), \sin(\theta_x)), (\cos(\theta_y), \sin(\theta_y))$.  We have
\[
(\sigma_x- \sigma_y)^2 \asymp [\nabla_e\theta)]^2_{2\pi}.
\]

Given a region $R \subset \Z^d$ which is finite and any $x \in R$, suppose that we represent a spin $\sigma_x= (\cos(\theta_x), \sin(\theta_x))$.  We will find it revealing to make the change of variables
\begin{align*}
\theta'_x=& \theta_x-g'_x\\
& \text{ where } g'_x = \cos(\theta_x) \gld_{R, x}\\
&\text{ or } g'_x = \cos(\theta_x) \gln_{R, x}
\end{align*}
depending on the context.  Note that despite the fact that $\theta_x$ is only defined modulo $2\pi$ this transformation is unambiguous.  Moreover, it is nonsingular as long as $|g_x|< 1$.

\begin{lemma}
\label{L:BLayer}
Let us suppose that $R \subset \Z^d$ is fixed and bounded. Then for any $\lambda \geq 0$,
\begin{multline}
\label{E:BL}
 -\HH_{R}(\sigma| \eta)
 = \sum_{e \cap R \neq 0} [\cos(\nabla_e\theta')-1] + \frac{1}{4}\sum_{x \in R} m_x \cos^2(\theta'_x) \\
 + \sum_{\stackrel{y \in \pO R}{x \sim y, x \in \pI R}} \gld_x \sin(\theta_y) + \frac 14 \sum_{\stackrel{y \in \pO R}{x \sim y, x \in \pI R}} \cos^2(\theta_y)[\gld_y-\gld_x]^2
  + Error^{\lambda, D}_{R}(\sigma)
\end{multline}
where
\[
m_x= \sum_{y \sim x} [\gld_y-\gld_x]^2
\]
and
\begin{equation}
Error^{\lambda, D}_{R} \lesssim \lambda \|\gld_R\|_{2} \| e_2\cdot \sigma\|_{2, R} + (\|\nabla \gld_R\|_{\infty}+ \| \gld_R\|_{\infty}) (\|\nabla \sigma\|_{2, R \cup \pO R}^2 +\|\nabla \gld_R\|_2^2).
\end{equation}

\vspace{10pt}
\noindent
Similarly, for free boundary conditions we have
\begin{equation}
\label{E:FL}
 -\HH_{R}(\sigma)
 = \sum_{e \subset R} [\cos(\nabla_e\theta'))-1] + \frac{1}{4}\sum_{x \in R} m_x \cos^2(\theta'_x)
  + Error^{\lambda, N}_R
\end{equation}
where
\[
m_x= \sum_{{\stackrel{y \sim x}{\langle x, y\rangle \subset R}}} [\gln_y-\gln_x]^2
\]
and
\begin{multline}
Error^{\lambda, N}_{R}\lesssim\\ \epsilon |\alpha_{R}|\|\sigma \cdot e_2\|_{1, R}+ \lambda \|\gld_R\|_{2} \| e_2\cdot \sigma\|_{2, R} + (\|\nabla \gln_R\|_{\infty}+ \| \gln_R\|_{\infty}) (\|\nabla \sigma\|_{2, R}^2 +\|\nabla \gln_R\|_2^2).
\end{multline}
\end{lemma}
\begin{proof}[Proof of \cref{L:BLayer}]
In the proof we use the notation $\bar{f}_e=\frac{f_x+ f_y}{2}$ where $e= \langle x, y \rangle$ for any (vector valued) function $f$ on vertices.  

We prove only the first estimate.
Let $g'_x = \cos(\theta_x) \gld_x$ and note that $g'_x$ is independent of the branch chosen in the definition of $\theta_x$. 

First, for any edge $e$ with $e\cap R \neq \varnothing$ choose $\theta_x$ and $\theta_y$ which achieve $[\nabla_e\theta)]^2_{2\pi}$. This choice is, in general, edge dependent.  Letting $\theta'_x= \theta_x-g'_x$,
\begin{equation}
\label{E:51}
\cos(\nabla_e\theta)=
\cos(\nabla_e\theta') - \sin(\nabla_e\theta')\nabla_e g'
-\frac{1}{2}\cos(\nabla_e \theta')(\nabla_e g')^2 + O([\nabla_e g']^3).
\end{equation}

For a positively oriented edge $e= \langle y, x \rangle$, we will use the formulas
\begin{align*}
&\theta_x= \bar{\theta}_e+ \frac{1}{2} \nabla_e\theta\\
&\theta_y= \bar{\theta}_e- \frac{1}{2} \nabla_e\theta.
\end{align*}
Using them, we may write
\begin{equation}
\label{E:52}
\nabla_eg'=  \cos(\bar{\theta}_e)\nabla_e\gld + O(|[\nabla_e\theta]_{2\pi}|\overline{|\gld|}_e + [\nabla_e\gld]^2)
\end{equation}

Combining \cref{E:51} and \cref{E:52} and summing over edges
\begin{multline}
\label{E:Int1}
\sum_{e \cap R \neq \varnothing} \cos(\nabla_e\theta)= \sum_{e \cap R \neq \varnothing } \cos(\nabla_e\theta') - \cos(\bar{\theta}_e)\nabla_e\theta \nabla_e\gld +\frac{\cos^2(\bar{\theta}_e)}{2}(\nabla_e\gld)^2
\\
+O((\|\nabla \gld_R\|_{\infty}+ \| \gld_R\|_{\infty}) (\|\nabla \sigma\|_{2, R \cup \pO R}^2 +\|\nabla \gld_R\|_2^2)).
\end{multline}

Now consider the contribution to $-\HH_R$ coming from the random field.
We have
\[
\epsilon \sum_{x \in R} \alpha_x \sin(\theta_x)= \sum_{x \in R} (-\Delta^D_{R} +\lambda)\cdot\gld_x \sin(\theta_x).
\]
By Cauchy-Schwarz, we have
\begin{equation}
\label{E:Int3}
|\sum_{x \in R} \lambda \gld_x \sin(\theta_x)| \leq \lambda \|\gld\|_{2, R} \|\sin(\theta)\|_{2, R}.
\end{equation}
Making a summation-by-parts,
\begin{equation}
\label{E:Int2}
\sum_{x \in R} -\Delta_D\cdot\gld_x \sin(\theta_x)= \sum_{e \cap R \neq \varnothing} \nabla_e\gld\nabla_e\sin(\theta) + \sum_{\stackrel{y \in \pO R}{x \sim y, x \in \pI R}} \gld_x \sin(\theta_y).
\end{equation}
Focusing on the first term in the RHS of \cref{E:Int2}, for each edge $e= \langle x, y \rangle$ we choose the same branches of $\theta_x, \theta_y$ as in \cref{E:Int1} and expand around $\bar{\theta}_e$:
\begin{equation}
\label{E:Int4}
\sum_{e \cap R \neq \varnothing} \nabla_e\gld\nabla_e\sin(\theta)
=\sum_{e \cap R \neq \varnothing} \cos(\bar{\theta}_e) \nabla_e\gldÊ\nabla_e\theta + O(|\nabla_e\gld |[\nabla_e\theta]^2_{2\pi})
\end{equation}
If -- edge by edge -- we combine the RHS of \cref{E:Int1,E:Int2,E:Int4} we obtain \cref{E:BL} with second term replaced by
\[
\sum_{e \cap R \neq \varnothing} \frac{\cos^2(\bar{\theta}_e)}{2}(\nabla_eg)^2.
\]
 \Cref{E:BL} is obtained by first replacing
$\frac{\cos^2(\bar{\theta}_e)}{2}$ with $\frac 12 \left[\frac{\cos^2(\theta_x)}{2}+ \frac{\cos^2(\theta_y)}{2}\right]$ at the cost of a term bounded by the error already amassed and noting
\[
\sum_{x \in  R}\frac{\cos^2(\theta_x)}{4} m_x =\frac 14 \sum_{x \in  R}\cos^2(\theta'_x) m_x + O\left(\sum_{e \in R} \overline{|g|}_e[\nabla_eg]^2\right).
\]
\end{proof}

\section{Preparatory Lemmas for \cref{SS:7}}
\label{S:Aux}

\subsection{Absence of Defects:  Proof of \cref{C:LocRef,,C:E1}}
\label{S:AuxBound}
In this section we consider the extent to which the Dirichlet energy $\EE_{Q_{L_0}}(\sigma)$ can be used to control `smoothness' of spin configurations in cubes $Q_{L_0}$ (these need not be $L_0$-measurable.  We consider this question in $\Z^d$ for $d \geq 1$.

\begin{definition}
\label{d:defect}
Let  $\mu \in (0, 1)$ and let $0<\delta< \mu$ be fixed.  We shall say that $\sigma\in \mathcal S_{Q_{L_0}}$ has a \textrm{defect} (w.r.t. $\delta, \mu$) if $\sigma(Q_{L_0}) \cdot e_1 \geq 1-\delta$ but there exists a graph connected subset $R \subset  Q_{L_0}$ so that $\diam(R) \geq \frac{L_0}{4}$ and $\dist(R, Q_{L_0}(r)^c) \geq \frac{L_0}{2}$ and so that for all $x \in R$, $\sigma_x \cdot e_1  \leq 1-\mu$ and \end{definition}

We wish to argue that defects cannot occur for low (Dirichlet) energy configurations for $d\in\{2, 3\}$. The following lemma depends on low dimensionality.  In \cref{S:Collar} it
 allows us, at the boundary of clean contours, to restrict the Hamitonian to regions of the phase space on which it is (morally) convex.

 Recall $\ell \sim \epsilon^{-1}|\log \epsilon|^{-4}$.  A similar statement can be pushed through in the two dimensional case.
\begin{lemma}[Energetic Cost of Defects]
\label{L:Defect}
Let $d=3$ and fix $0< \delta< \mu < 1$. There exists $\epsilon \in (0, 1)$ so that if $\EE_Q(\sigma) \leq \epsilon^2|\log \epsilon| \ell^3$, $\sigma$ does not have defects in $Q_{\ell}$.
\end{lemma}
This lemma is an easy consequence of the following.  Let $B_l$ denote a fixed $\ell^2$ ball of radius $l$ in $\Z^d$.
\begin{lemma}[Point Defects in $d$-Dimensions]
\label{L:PD}
Let $d \geq 1$.  If $\mu \in (0, 1)$ and $\delta \in (0, \mu)$ then for all $l$ sufficiently large, whenever  $\sigma(B_{l}) \cdot e_1 \geq 1- \delta$ and $\sigma_0\cdot e_1 \leq 1- \mu$
\[
\EE_{B_{l}}(\sigma) \gtrsim
\begin{cases}
 \frac{1}{ l} \text{ if $d=1$},\\
 \frac{1}{\log l} \text{ if $d=2$},\\
 1 \text{ if $d\geq3$},\\
 \end{cases}
\]
where the implicit constant depends on $\delta, \mu$ and dimension.
\end{lemma}

\begin{proof}[Proof of \cref{L:Defect} given \cref{L:PD}]
Let us assume for convenience that $Q_{\ell}$ is centered at $0$, that is $Q_{\ell}= \{-\frac{\ell}{2}, \dotsc, \frac{\ell}{2}-1\}^d$.

In $\Z^3$ let $\{\hat e_1, \hat{e}_2, \hat e_3\}$ denote the usual basis of orthonormal, positively oriented unit lattice vectors.  Let $H_i(k)$ denote the hyperplane through $k\hat e_i$ perpendicular to $\hat e_i$ and let $\BB_i(k)= H_i(k)\cap Q_{\ell}$ for $k \in\{- \ffrac{\ell}{2}, \dotsc, \ffrac{\ell}{2}-1\}$. We may decompose $\sigma(Q_{\ell})=\ell^{-1} \sum_{k= - \ffrac{\ell}{2}}^{\ffrac{\ell}{2}-1} \sigma(\BB_i(k))$, i.e. as a sum over $d-1$ dimensional hypercubes perpendicular to each $\hat{e}_i$.  Setting $f_i(k)=  \sigma(\BB_i(k))$ we have
\begin{align*}
&\ell^{-1} \sum_{k} \|f_i(k) - \sigma(Q_\ell)\|^2 \leq \ell^{-3} \sum_{x \in Q_{\ell}}\|\sigma_x- \sigma(Q_\ell)\|_2^2 \text { and }\\
&\ell^{-3} \sum_{x \in Q_{\ell}}\|\sigma_x- \sigma(Q_\ell)\|_2^2 \underbrace{\ls}_{\textrm{I}}  \ell^{-1} \EE_{Q_{\ell}}(\sigma) \underbrace{\ls}_{\textrm{II}} \epsilon^2|\log \epsilon| \ell^2.
\end{align*}
Inequality $\textrm{I}$ follows from the P\'{o}incare inequality and $\textrm{II}$ from our hypothesis on the Dirichlet energy.  By Markov's inequality and the defintion of $\ell$, for each $i$
\begin{align}
&\ell^{-1}|\{k: \|f_i(k) - \sigma(Q_\ell)\|_2 \geq \ffrac{\mu-\delta}2 \}| \leq \ffrac{2}{\mu-\delta}|\log \epsilon|^{-7},\\
&\ell^{-1}|\{k: \EE_{\BB_i(k)}(\sigma) \geq |\log \epsilon|^{-6}\}| \leq |\log \epsilon|^{-1}.
\end{align}
Thus if $Q_\ell$ has a defect, we can find at least $\ffrac{\ell}{5}$ disjoint $2$ dimensional subcubes  $\BB_i(k)$ so that 
\begin{align}
&f_i(k)\cdot e_1 \geq 1- [\delta+ \ffrac{\mu-\delta}2],\\
&\EE_{\BB_i(k)}(\sigma) \leq |\log \epsilon|^{-6},\\
&\sigma_{x} \cdot e_1 \leq 1- \mu \text{ for some $x \in  \BB_i(k)$ so that $\dist(x, \pO Q_\ell) \geq \ffrac \ell4$}.
\end{align}
For each such $x$ let $b_x$ be the $\ell^2$ ball of radius $\ffrac \ell4$ around $x$ in $\BB_i(k)$.  Then
\begin{align*}
&\|\sigma(b_x)- f_i(k)\|^2_2 \ls \EE_{\BB_i(k))}(\sigma) \ls |\log \epsilon|^{-6},\\
&\EE_{b_x}(\sigma)\leq \EE_{\BB_i(k))}(\sigma).
\end{align*}
For $\epsilon$ sufficiently small, this is in contradiction with \cref{L:PD}.
\end{proof}

\begin{proof}[Proof of \cref{L:PD}]
Let us assume for convenience that $B_{l}$ is centered at $0$, i.e. $B_l= \{x\in \Z^d: \|x\|_2 \leq l\}$.
The proof of the Lemma is easiest when $d=1$:  The condition $\sigma(B_{l}) \cdot e_1 \geq 1- \delta$ implies that we can find a pair of vertices $x< 0< y$ such that $\sigma_x\cdot e_1, \sigma_y\cdot e_1 \geq 1- [\delta+ \ffrac{\mu-\delta}{2}]$ and $x, y \in [-l, l]$.  Optimizing the Dirichlet energy  $\EE$ in $[x, y]$ subject to these boundary conditions and the condition $\sigma_0 \cdot e_1< 1-\mu$ then gives
\[
\EE_{B_{l}}(\sigma) \gtrsim \ffrac{1}{l}
\]

For dimension $d\geq 2$ it is most convenient for us to replace $\sigma_x$ by a continuous interpolation.
For $y \in \R$ let $[y]$, $\{y\}$ denote the integer and fractional parts of $y$ respectively and extend these notations componentwise to vectors in $\R^d$.  Let $\tilde{B}_{k}=\{ x \in \R^d: [x] \in B_{k}\}$. We define $\tilde \sigma$ on 
$\tilde{B}_{l-2}$ by
\[
\tilde{\sigma}(x)=\sum_{v \in \{0, 1\}^d} \left[ \prod_{i=1}^d(v_i\{x_i\} +(1-v_i)(1-\{x_i\}))\right] \sigma_{[x] +v}
\]
i.e. $\tilde{\sigma}(x)$ is a multi-linear interpolation of $\sigma^{\textrm{ext}}_x$.

By definition, we can find a small ball $B'$ (depending only on $d, \delta$ and $\mu$) around $0$ so that $\tilde{\sigma}(x)\cdot e_1< 1-  [\delta + \ffrac{3(\mu-\delta)}{4}]$ if $x \in B'$.
In addition,
\begin{equation}
\label{E:Mean}
|\tilde B_{l-5}|^{-1}\int_{\tilde{B}_{l-5}} \tilde{\sigma}(x)\cdot e_1 \textrm{d}^d x = |B_{l}|^{-1} \sum_{x \in B_{l}} \sigma \cdot e_1+ O( \ffrac{1}{l})
\end{equation}
with the second term on the RHS coming from boundary integrations.  Taking $l$ large enough, we may assume $|\tilde B_{l}|^{-1}\int_{\tilde{B}_{l}}\textrm{d}^d x \tilde{\sigma}(x)\cdot e_1 \geq 1-  [\delta + \ffrac{(\mu-\delta)}{4}]$.
Further,
there exist universal constants $c_1, c_2>0$ so that
\[
c_1\sum_{e \subset B_{l-5}} [\nabla_e \sigma]^2 \leq \int_{\tilde{B}_{l-5}}\textrm{d}^d x |\nabla \tilde{\sigma}_x|^2 \leq c_2 \sum_{e \subset B_{l-3}} [\nabla_e \sigma]^2
\]
where $\nabla \tilde{\sigma}_x$ is the gradient of $\tilde{\sigma}_x$ in $\R^d$.
It is therefore enough to obtain bounds in the continuum setting with $|\sigma_x|=1$ relaxed to $|\tilde{\sigma}_x|\leq 1$.

We assume $\tilde{\sigma}_x \cdot e_1 \geq 0$ and that $\tilde{\sigma}_x$ is spherically symmetric (although perhaps not of unit length) since the operation of taking spherical averages preserves \cref{E:Mean} and can only lower the kinetic energy.  

Now we proceed in a manner similar to the one dimensional case.
We can find a sphere
\[
S_r=\{x\in \R^d: \|x\|_2=r\}
 \]
 with $ r \leq l$ so that $\tilde{\sigma}(\|x\|_2)\cdot e_1> 1-  [\delta + \ffrac{\mu-\delta}{2}].$

Combined with the condition $\tilde{\sigma}(x)\cdot e_1 \leq  1-  [\delta + \ffrac{3(\mu-\delta)}{4}]$ if $x \in B'$ and using capacitance estimates between balls in $\R^d$, this implies
\[
\int_{\tilde{B}_{l}}\textrm{d}^d x |\nabla \tilde{\sigma}_x|^2 \gtrsim
\begin{cases}
\frac{1}{1+\log (l)} \text{ for $d = 2$}\\
1 \text{ for $d \geq 3$}
\end{cases}
\]
The implicit constants depend on the capacitance between a unit ball and one of radius two and also on $\delta, \mu$.
\end{proof}

\begin{proof}[Proof of \cref{C:LocRef,,C:E1}]
These are easy consequence of \cref{L:Defect}
\end{proof}

\subsection{Localization Estimates for $\gld_R$}
Given a region $R$, in this section we show how to control the field  $\gld_{R}$ in terms of the family of fields $(\gld_{Q})_{Q}$ where the boxes $Q$ are of the form $Q= Q(r + \eta)$ for some $Q(r) \in \QQ_L$ and some $\eta \in N: =\frac{L}{16}\{ -32\dotsc, 32\}^3$.  Here, as in \cref{S:Collar}, $\lambda=
L^{-2} \log^{8} L$.

Recall that in \cref{S:Contours}, we defined the function $R$ on $\QQ_L$ by $R(Q)= \max_{\eta} \|\alpha\|_{\infty, Q_{\eta}}$.  Recall that $\QQ_L'= \{ Q_{\eta}: Q\in \QQ_{L} \text{ and } \eta \in \NN\}$.  Here is a convenient, though not particularly optimal, bound which we use later:
\begin{lemma}
\label{L:TechBadRandomeness}
Let $\tilde{Q}\in \QQ_L'$ so that $\tilde{Q}\cap R \neq \varnothing$.
Suppose there exists $\delta>0$ so that $x \in \tilde{Q} \cap R$ and $\dist(x, \pO[\tilde{Q} \triangle \mathscr D]) \geq \lambda^{-\frac 12}  |\log \epsilon|^{\delta}$.  Then there exists $c>0$  (independent of $\delta$) so that
\[
|\gld_{\tilde{Q}, x}-\gld_{R, x}| \ls \epsilon \lambda^{-1} \sum_{\stackrel{Q\cap R \neq \varnothing}{Q\in \QQ_L}} R(Q)\left[e^{-c |\log \epsilon|^{\delta}}\wedge  e^{-c \lambda^{\frac 12} \dist(x, Q}\right].
\]

In particular, if $\sup_{Q\cap R \neq \varnothing} R(Q) \leq K$ and $\delta>1$, 
\begin{equation}
\label{E:TechBadRandomeness1}
|\gld_{\tilde{Q}_L, x}-\gld_{R, x}| \ls Ke^{-c |\log \epsilon|^{\delta}}.
\end{equation}
\end{lemma}
Let us emphasize here that $\pO[R_1 \triangle R_2]$ denotes the outer boundary of $R_1 \cup R_2\backslash R_1 \cap R_2$.  In particular, it is possible for $x\in R_1 \cap R_2$ to satisfy the hypotheses of this Lemma while being close to the boundary, as long as the part of the boundary it is close to is shared by both $R_1, R_2$.
\begin{remark}
One annoying problem if $d=2$ is that this bound cannot be used as the restrictions on $\ell,L$ are too tight to take $\delta>1$. 
\end{remark}

\begin{proof}
The proof exploits a random walk representation and coupling argument.  Set $R_1=R, R_2 =\tilde{Q}$.
Let $X^{(i)}_t$ for $i \in \{1, 2\}$ denote continuous time simple random walks started from $x$.  In order for things to work out correctly, we take our walks to have exponential holding times of rate $2d$.  For each $i\in \{1, 2\}$, let $\tau_{i}$ denote the exit time of $X^{(i)}$ from $R_i$ and let $\tau_{\lambda}$ be an independent exponential random variable which has rate $\lambda$.  We couple $X^{(1)}_t, X^{(2)}_t$ until the first time the coupled walk hits $R_1 \triangle R_2$.  Call this hitting time $\tau_3$ (it can be infinite).
Then
\begin{align*}
\gld_{R_i, x}&= \epsilon\E_x\left[\int_0^{\tau_{R_i} \wedge \tau_{\lambda}}\alpha_{X^{(i)}_t}\right],\\
\left| \gld_{R_1, x}- \gld_{R_2, x}\right|& =\epsilon \left|\underbrace{\E_x\left[\int_0^{\tau_{1}}\alpha_{X^{1}_t} \mathbf 1\{\tau_{1}< \tau_{\lambda}, \tau_3\}\right]}_{\textrm{I}}-\underbrace{\E_x\left[\int_0^{\tau_{2}}\alpha_{X^{2}_t} \mathbf 1\{\tau_{2}< \tau_{\lambda}, \tau_3\}\right]}_{\textrm{II}} \right|.
\end{align*}

Clearly
\[
\textrm{II} \leq R(Q') \E[\tau_{2}\mathbf 1\{\tau_{2}< \tau_{\lambda}, \tau_3\}]
\]
for any $Q' \in \QQ_L$ which intersects $\tilde{Q}$. 
We claim
\begin{equation}
\label{E:bounda}
\E[\tau_{2}\mathbf 1\{\tau_{2}< \tau_{\lambda}, \tau_3\}] \ls \lambda^{-1} e^{-c|\log \epsilon|^{\delta}}.
\end{equation}
Note that on the event $\{\tau_{2}< \tau_3\}$, $|X^{(i)}_{\tau_{2}}-x| \geq \lambda^{-\frac 12} |\log \epsilon|^{\delta}$.   Thus, if $r(\lambda)= \ffrac 12 \lambda^{-\frac 12} |\log \epsilon|^{\delta}$ and $\tau_{r(\lambda)}$ denotes the first exit time from the $\ell^{\infty}$ ball of radius $r(\lambda)$ around $x$, we can bound the expected value by
\[
\E[\tau_{2}\mathbf 1\{\tau_{2}< \tau_{\lambda}, \tau_3\}]\leq \E[\tau_{\lambda}\mathbf 1\{\tau_{r(\lambda)}< \tau_{\lambda}\}].
\]
The inequality \cref{E:bounda} now follows.  

Term \textrm{I} is estimated similarly. Let $\tau_Q$ be the hitting time of $Q \in \QQ_L$ for $X^{(1)}$.  Then
\[\textrm{I} \leq \sum_{Q \cap R \neq \varnothing}  R(Q) \E_x\left[ \tau_{1} \mathbf 1\{\tau_Q<\tau_{1}< \tau_{\lambda}, \tau_3\}\right].
\] 
Reasoning as above,
\[
 \E_x\left[ \tau_{1} \mathbf 1\{\tau_Q<\tau_{1}< \tau_{\lambda}, \tau_3\}\right] \ls \lambda^{-1}e^{-c |\log \epsilon|^{\delta}}\wedge  e^{-c \lambda^{\frac 12} \dist (x, Q)}
\]
from which the claim follows.  
Note that we are being quite generous here as there should be cancellations which reduce \textrm{I} and \textrm{II}.  
\end{proof}

\subsection{Estimates for Maximizers of $\KK_R$}
\label{S:AuxRelax}
In this subsection we consider the behavior of optimizers of the functional 
\begin{equation}
\label{E:AuxHam1}
\KK_{\mathfrak g^+}(\varphi|\phi)=\sum_{e \cap \mathfrak g^+ \neq 0} \cos(\nabla_e\varphi)-1 + \frac 14\sum_{x\in  \mathfrak g^+} m_x \cos^2(\varphi_x)
\end{equation}
where $m_x= \sum_{y \sim x} [\gld_{\mathscr C^+, y}-\gld_{\mathscr C^, x}]^2$.
To summarize, the properties we will use below are:
\begin{align*}
&\|\phi\|_{\infty, \pO \mathfrak g^+} \leq \ffrac{\pi}6\\
&\mathfrak g^+ \subset \cup_{\Xi_L(Q)=1} Q \\
&\phi|_{\pO \Lambda_N \cup [\mathfrak D^+_{\ffrac L{12}} \cap \pO \mathfrak g^+]}=0.
\end{align*}
\begin{lemma}
\label{L:MaxPrince}
Suppose that $R \subset \mathfrak g^+$ is bounded set.
Consider the functional
\[
\KK_R(\varphi|\tau)
\]
with a boundary condition $\tau$ satisfying $\|\tau\|_{\infty, \pO R} \leq \ffrac \pi 6$.  Then $\KK_R$ has a unique maximizer $(\vartheta_x)_{x \in R}$
and
\[
\|\vartheta\|_{\infty, R} \leq \|\tau\|_{\infty, \pO R}.
\]
\end{lemma}

\begin{proof}[Proof of \cref{L:MaxPrince}]
Clearly we may assume $\vartheta_x \in [-\ffrac \pi2, \ffrac \pi2]$, otherwise the configuration cannot maximise $\KK_R$.
Suppose first that $\vartheta_x$ is a maximizer and that
 $\vartheta_x \in (-\ffrac{\pi}{2}, \ffrac{\pi}{2})$ for all $x \in R$.

We may think of the equation satisfied by a stationary point of $\KK_R$ in the form 
\begin{equation}
\label{E:Stat}
-\sum_{y \sim x} {\tt C}_{xy}[\vartheta_x- \vartheta_y] + V_x \vartheta_x = 0.
\end{equation}
where the conductances and potentials are given in terms of $\vartheta$ by
\begin{align*}
&{\tt C}_{e} = \frac{\sin(\nabla(e, \vartheta))}{\nabla(e, \vartheta)} \text{ for all $e \cap R \neq \varnothing$},\\
&V_x= -m_x \cos(\vartheta_x)\frac{\sin(\vartheta_x)}{2\vartheta_x} \text{ for all $x \in Q_{L}$}.
 \end{align*}
Because we are assuming $\vartheta_x \in (-\ffrac{\pi}{2}, \ffrac{\pi}{2})$ for all $x \in R$, ${\tt C}_{e} > 0$  $V_x< 0$.

Thinking of ${\tt C}_e, V_x$ as given and fixed,
 it is natural to define corresponding \textit{linear} operators on functions $f$ on $R \cup \pO R$.
\begin{align*}
&-L^{{\tt C}}\cdot f_x = \frac 1 {{\tt C}_x}\sum_{y \sim x} {\tt C}_{xy}[f_x- f_y] \text{ for $x \in R, y \in R \cup \pO R$}\\
&M^V\cdot f_x= \frac{|V_x|}{{\tt C}_x} f_x \text{ for $x \in R$}
\end{align*}

We may rewrite \cref{E:Stat} as
\begin{equation}
\label{E:Stat1}
(-L^{{\tt C}}+ M^V)\cdot f_x=0
\end{equation}
subject to  the boundary condition $f|_{\pO R} \equiv \tau$.
The operator $(-L^{{\tt C}}+ M^V)^{-1}$ is positivity preserving, so if $f^{\tau'}$ denotes the solution with boundary condition $\tau'$ then
\begin{align*}
-f^{|\tau|}\leq f^{\tau} \leq f^{|\tau|} \text{ and}
\| f^{|\tau|}\|_{\infty, R} \leq \|\tau\|_{\infty, \pO R}.
\end{align*}
Hence, if we can show that maximizers do not take on the values $\{-\ffrac \pi2, \ffrac \pi2\}$ the lemma will be proved
as $\KK$ is uniformly convex in the region
\[
\theta_x \in [-\ffrac \pi 6, \ffrac \pi6] \: \: \forall x \in R.
\]

To show this let us in introduce a vector field on $[-\ffrac \pi2, \ffrac \pi2]\times R$ of the form
\[
X(\theta, x) = \begin{cases}
-[\theta-\ffrac{\pi}{6}] \quad \text{ if $\theta> \ffrac{\pi}{6}$},\\
-[\theta + \ffrac{\pi}{6}] \quad \text{ if $\theta< \ffrac{-\pi}{6}$},\\
0 \text{ otherwise}.
\end{cases}
\]
and consider the flow $\partial_t \theta(t, x):= X(\theta(t, x), x)$.  Then from any initial configuration $\theta_x(0)\in [-\ffrac{\pi}{2}, \ffrac{\pi}2]$ we find that $\KK(\theta(t)|\tau)$ is increasing, strictly so if there is a space time point $(t, x)$ so that $|\theta(t, x)| > \ffrac{\pi}{6}$.
Our claims then follow immediately.
\end{proof}

We next claim that the maximizer $\vartheta_x$ of $\KK_{\mathfrak g^+}(\theta|\tau)$ is very small for any $x\in \mathfrak g^+$ so that either $\dist(x,  \pO \mathfrak g^+)$ is (much) larger than $\epsilon^{-1}$ or $\dist(x,  \textrm{M} \cap[\pi \Lambda_N \cup \mathfrak D^+_{\ffrac{L}{12}}])$ small enough.
\begin{lemma}[Bulk Behavior]
\label{L:Relax}
There exists $\epsilon_0 \in (0, 1)$ so that if $0<\epsilon< \epsilon_0$ the following holds.  Let $\vartheta$ optimize
\[
 \KK_{\mathfrak g^+}(\vartheta|\tau) \text{ with }
\|\tau_x\|_{\infty, \pO \mathfrak g^+}  \leq \ffrac{\pi}{6} .
\]
Then for any $x \in \mathfrak g^+$ so that $\dist(x, \pO \mathfrak g^+)\geq \ffrac L8$
\[
0 \leq |\vartheta_x| \ls \exp\left[-c\log^4 \epsilon \right].
\]
\end{lemma}

\begin{lemma}[Boundary Behavior]
\label{L:BBehavior}
There exists $\epsilon_0 \in (0, 1)$ so that if $0<\epsilon< \epsilon_0$ the following holds.   Let $\vartheta$ maximize $ \KK_{\mathfrak g^+}(\vartheta|\tau)$
subject to a boundary condition $\tau$ with
\[
\|\tau_x\|_{\infty, {\pO \mathfrak g^+}} \leq \ffrac{\pi}{6} \text{ and }
\tau_x=0
\text{ for all $x\in \pO \mathfrak g^+ \cap [\pO \Lambda_N \cup \mathfrak D^+_{\ffrac{L}{12}}]$}.
\]
Then
\[
0 \leq |\vartheta_x| \ls \exp\big[-c\log^4 \epsilon \big]\\
\]
whenever $\dist(x, \textrm{M}^{+} \cap  [\pI \Lambda_N \cup \mathfrak D^+_{\ffrac{L}{12}}]) < \ffrac{L}{4}$.
\end{lemma}

Together these Lemmas imply
\begin{corollary}
\label{C:Contract}
There exists $\epsilon_0 \in (0, 1)$ so that if $0<\epsilon< \epsilon_0$ the following holds.  If $\dist(x, M) \leq \ffrac{L}{5}$ then
\[
0 \leq |\vartheta_x|
\ls
\exp\big[-c\log^4 \epsilon \big]
\]
\end{corollary}

The proofs require the following technical input.  Let $X_t$ be a continuous time random walk on $\Z^d$ with conductances $(\textrm{C}_e)_{e \in \EE}$ and let $\Q_x^{\textrm{C}}$ be the probability measure for the process.
\begin{theorem}[Gaussian Bounds \cite{Delmotte}]
\label{T:Gaussian}
Suppose there exist constants $A, B>0$ so that the conductances $A<\tt C_e<B$  for all $e$.  Then there exist constants $c_l, C_l, c_g, C_g$ so that for all $t> 1$ and all $x \in \Z^d$
\begin{equation}
C_l t^{-\ffrac d2} e^{-c_l \ffrac{\|x-y\|_2^2}{t} \wedge(1+ \|x-y\|_2)}\leq \Q^{\tt C}_x[ X_t =y] \leq C_g t^{-\ffrac d2} e^{-c_g \ffrac{\|x-y\|_2^2}{t} \wedge(1+ \|x-y\|_2)}
\end{equation}
\end{theorem}

\begin{proof}[Proof of \cref{L:Relax}]
The proof uses the Feynman-Kac formula and Azuma's inequality. Let $\vartheta$ be the unique maximizer of $\KK_{\mathfrak g^+}$ and set
\begin{align*}
&{\tt C}_{e} = \frac{\sin(\nabla(e, \vartheta))}{\nabla(e, \vartheta)} \text{ for all $e \cap \mathfrak g^+ \neq \varnothing$},\\
&V_x= -m_x \cos(\vartheta_x)\frac{\sin(\vartheta_x)}{2\vartheta_x} \text{ for all $x \in {\mathfrak g^+}$}.
 \end{align*}
By \cref{L:MaxPrince}, there are constants $A, B>0$ so that $A<\ffrac{|V_x|}{m_x},  \texttt{C}_e< B$.  In order to be able to apply the bounds of \cref{T:Gaussian}, we extend $\tt{C}_e=1$ for all $e$ such that $e \cap \mathfrak g^{+}= \varnothing$.

Let $X_t$ be a continuous time Markov chain with conductances $\tt{C}_e$ started from $x \in \mathfrak g^{+}$, $\Q_x$ denoting expectation of the corresponding measure.  Let
\[
R= \dist(x, \pO \mathfrak g^+) \wedge \ffrac L2
\]
By assumption $R> \ffrac L8$.
We define the stopping time
$\upsilon_{x}(\ffrac 1k)$ as the first exit time of $X_t$ from $\{y: \|x-y\|_{\infty} \leq \ffrac{R}{k}\}$ and set $\upsilon_{\mathfrak g^{+}}$ to be the first exit of $X_t$ from $\mathfrak g^{+}$.
Using the Feynman-Kac representation and the \textit{ a priori} bound on $|V_x|$,
\begin{equation}
|\vartheta_x|= \left|\mathbb Q_x\left[e^{-\int_{0}^{\upsilon_{\mathfrak g^+}} |V_{X_t}| \textrm d t}\tau_{X_{\upsilon_{\mathfrak g^+}}}\right] \right| \leq \mathbb Q_x\left[e^{-A\int_{0}^{\upsilon_{x}(1/2)} m_{X_t} \textrm d t}\right] \|\tau\|_{\infty, \pO \mathfrak g^{+}}.
\end{equation}

Let $\mathcal F_{t}$ denote the natural filtration of $X_{t \wedge \upsilon_x(1/2)}$.  For any sequence of times
\[
s_0=0< s_1< \dots < s_J< \dist(x, \pO \mathfrak g^{+})^2=s_{J+1},
\]
let
\[
\mathfrak I_{i} = \int_{s_{i}\wedge \upsilon_{x}(1/2)}^{s_{i+1}\wedge \upsilon_x(1/2)} m_{X_t}\textrm d t- \Q_x\left[\int_{s_{i}\wedge \upsilon_x(1/2)}^{s_{i+1}\wedge \upsilon_x(1/2)} m_{X_t}\textrm d t|\mathcal F_{s_{i}}\right]
\]
and note that
\[
\int_{0}^{\upsilon_x(1/2)}  m_{X_t}\textrm d t\geq \sum_{i=0}^{J+1} \mathfrak I_i + \sum_{i=0}^J \mathbf 1\{s_{i}< \upsilon_x(1/2) \}\Q_{X_{s_{i}}}\left[\int_{0}^{s_{i+1}-s_{i}} m_{X_t}\textrm{d} t \mathbf 1\{s_{i+1}-s_i< \upsilon_x(1/2) \}\right]
\]
where the strong Markov property has been used to evaluate the second term and we note that the first sum is of martingale increments. We choose $s_{i+1}-s_{i}\equiv \log^{180}(L)$ for all $i$.

Recall that $\QQ_L'$ denotes the set of shifted cubes $\{Q(r+ \eta): Q\in \QQ_L, \eta\in \NN\}$.
For each $x \in \mathfrak g^{+}$, we can find a cube $Q$ of sidelength $L$ so that 
\begin{align*}
&Q\in \QQ_L' \text{ and }\Xi_L(Q)=1,\\
&\left\{y: \|x-y\|_{\infty} \leq R\right\} \subset Q.
\end{align*}

Consider the increment
\[
\textrm{I}=\Q_{z}\left[\int_{0}^{s_{i+1}-s_{i}} m_{X_t}\textrm{d} t  \mathbf 1\{s_{i+1}-s_i< \upsilon_x(1/2)\} \right] 
\]
By  \cref{L:TechBadRandomeness} and because $\mathscr C^+ \subset \cup_{\Xi_L(Q')=1}Q'$
\[
m_{X_t}= \sum_{y \sim X_t} [\nabla_{\langle X_t y \rangle} \gld_Q]^2 + O(e^{-c \log^4 \epsilon}).
\] 
Because $\Xi(Q)=1$ the event $\AA_{\log^{90} L}$ occurs in the box $Q$.  For $\epsilon$ small enough, the Gaussian bounds in \cref{T:Gaussian} imply
\[
\textrm{I} \gtrsim \epsilon^2 [s_{i+1}-s_{i}].
\]
 as long as $z \in\{ y: \|x-y\|_\infty < \ffrac R2\}$.
 
We have
\begin{multline}
\sum_{i=0}^J \mathbf 1_{\{s_{i}< \upsilon_x(\ffrac 12)\}}\Q_{X_{s_{i}}}\left[\int_{0}^{s_{i+1}-s_{i}} m_{X_t}\textrm{d} t  \mathbf 1\{s_{i+1}-s_i< \upsilon_x(\ffrac 12)\}\right] \\
\gtrsim  \sum_{i=0}^J \mathbf 1\left\{s_{i}< \upsilon_x(\ffrac 12), X_{s_i}\in \{y: \|x-y\|_{\infty} \leq \ffrac R2\right\} \epsilon^2[s_{i+1}-s_{i}].
\end{multline}
Let $M \geq 2$ be fixed.  Introducing the event $F=\{\upsilon_{x}(1/4) > \ffrac{R^2}{M}\}$, we thus obtain
\[
\mathbb Q_x\left[e^{-\int_{0}^{\upsilon_x(1/2)} m_{X_t} \textrm d t}\right]\ls e^{-c_1 \epsilon^2 \ffrac{R^2}{M}}\mathbb Q_x\left[e^{-\sum_{i=0}^{J+1} \mathfrak I_i}\mathbf 1_{F}\right] + e^{-c_2M}.
\]
Here $c_1,c_2$ are universal constants and the last term comes from applying the Gaussian estimate of \cref{T:Gaussian} to the event $F^c$.

Using \cref{L:TechBadRandomeness} once again, 
\[
|\nabla(e, \gld_{\mathscr C^+})| \ls \epsilon \log^{30} \epsilon,
\]
so that
\[
|\mathfrak I_i|\ls \epsilon^2 |\log \epsilon|^{240}
\]
where we used $\log L \sim |\log \epsilon|$.

Because the $\frak I_i$ are martingale increments, Azuma's Inequality then provides the bound
\begin{align*}
\mathbb Q_x\left[e^{-\sum_{i=1}^{J+1} \mathfrak I_i}\mathbf 1_{F}\right] & \leq \mathbb Q_x\left[e^{-\sum_{i=0}^{J+1} \mathfrak I_i}\right]\\
& \ls e^{c J [\epsilon^2 \log^{240}(\epsilon)]^2 }\\
& \ls e^{c R^2 \epsilon^{3}}.
\end{align*}

Gathering the estimates together, we have proved that
\begin{equation}
0 \leq \vartheta_x \ls \exp\left[c_1\left(-\ffrac{\epsilon^2}{M}  + \epsilon^{3}\right)L^2\right] +\exp \left[-c_2M\right].
\end{equation}
For $\epsilon$ small enough, optimizing with respect to $M$ proves the Lemma.
\end{proof}

\begin{proof}[Proof of \cref{L:BBehavior}]
We begin with the following observation.  Using notation from the previous proof, suppose that $x \in \mathfrak g^+$ and $\dist(x,  \textrm{M}\cap[\pI \Lambda_N \cup \mathfrak D^+_{\ffrac L{12}}]) < \ffrac L4$.

We may write
\begin{equation}
|\vartheta_x| \leq {\mathbb Q_x\left[e^{-A \int_{0}^{\upsilon_{\mathfrak g^{+}}} m_{X_t} \textrm d t}|\tau_{X_{\upsilon_{\mathfrak g^{+}}}}|\right]}
\end{equation}
where $\upsilon_{\mathfrak g^{+}}$ denotes the first exit time of $\mathfrak g^+$ and $A\leq \ffrac{|V_x|}{m_x}$.  If we let
\[
E=\{ X_{\upsilon_{\mathfrak g^{+}}} \notin [\pO \Lambda_N \cup\mathfrak D^+_{\ffrac{L}{12}}]\},
\]
then because we have assumed that the boundary condition $\tau_x= 0$ for all $x \in \pO \Lambda_N \cup~[\mathfrak D^+_{\ffrac{L}{12}} \backslash \mathfrak g^+]$,
\[
|\vartheta_x|\ls  \mathbb Q_x\left[e^{-A \int_{0}^{\upsilon_{\mathfrak g^{+}}} m_{X_t} \textrm d t}\: \mathbf 1_E \right] 
\]
Let $B_r(x)$ denote the $\ell^2$ ball of radius $r$ around $x$.  Then if  $x \in \mathfrak g^+$ and 
\[
\dist(x,  \textrm{M}\cap [\pO \Lambda_N \cup \mathfrak D^+_{\ffrac L{12}}]) < \ffrac L4
\] 
$B_{\ffrac L5}(x) \cap \pO \mathfrak g^+ \subset [\pO \Lambda_N \cup \pI \mathfrak D^+_{\ffrac L{12}}]$.  

If $\upsilon_{x}$ denote the first exit time from $B_{\ffrac L5}(x)$,
\[
|\vartheta_x| \ls\mathbb Q_x\left[e^{-A \int_{0}^{\upsilon_{x}} m_{X_t} \textrm d t} \mathbf 1_{E} \right] 
\]

However a small variation on the argument in the previous proof gives
\[
|\vartheta_x| \ls\mathbb Q_x\left[e^{-A \int_{0}^{\upsilon_{x}} m_{X_t} \textrm d t} \mathbf 1_{E} \right] 
\ls e^{-c |\log \epsilon|^4} 
\]
and so the bound follows immediately.
\end{proof}

\section{Proofs for \cref{S:Groundstates,,S:Collar,,S:Peierls}}

\subsection{Proofs for \cref{S:Groundstates}}
\label{SS:6}
In this section we will prove \cref{L:GSd3} of \cref{S:Groundstates}.  

\begin{proof}[Proof of \cref{L:GSd3}]
From our construction, the only statement which requires further justification is the first. 
Consider first the contributions from the interactions between neighboring boxes $Q_{\ffrac \ell2}(x_0), Q_{\ffrac \ell2}(x_1)\in \QQ_{\ffrac \ell2}(\Gamma)$.  For ANY edge $e=\langle x, y\rangle$ where  $x \in Q_{\ffrac \ell2}(x_0)$, $y\in Q_{\ffrac \ell2}(x_1)$ for $x_0 \neq x_1$,  $|\nabla_e \sigma^{\bar{\delta}(\Gamma)}|  \ls \epsilon |\log \epsilon|^{35}$.
Thus
\[
0 \leq \sum_{{\stackrel{e=\langle x, y \rangle:}{x \in Q_{\ffrac \ell2}(x_0), y \in Q_{\ffrac \ell2}(x_1)}}} 1- \sigma_x \cdot \sigma_y \ls 1.
\]
As a consequence, our choice of $\ell$ implies
\[
0 \leq  \sum_{Q\in \RR_{\ffrac \ell2 }(\Gamma)}  \textrm{E}_0(Q) + \HH_{\bar{\delta}(\Gamma)}(\sigma^{\bar{\delta}(\Gamma)}) \leq  C \epsilon^{\ffrac 52} \log^{30} |\Gamma|+  \sum_{Q \in \RR_{\ffrac \ell2}(\Gamma)}  \textrm{E}(Q) + \HH_{Q}(\sigma^{\bar{\delta}(\Gamma)}).
 \]
where we used the fact that $\textrm{E}(Q)\geq \textrm{E}_0(Q)$ on the RHS.

Fix $Q$ such that $\Xi_{\ffrac \ell 2}(Q)=1$.  In this case, a combination of \cref{L:FBE,,E:FL} implies that $\frac 12 \EE_{Q}(g^N_{Q})$ is (essentially) the maximum value $-\HH_{Q}$ can be on $\mathcal S_{Q}$.
Consider the free boundary condition maximum
\[
\textrm{E}(Q)= \max_{\sigma\in \mathcal S_Q}-\HH_Q(\sigma)
\]
We have
\begin{multline}
\left|-\HH_{Q}(\sigma^{\bar{\delta}(\Gamma)}) -\textrm{E}(Q)\right| \ls \underbrace{\left|-\HH_{Q}(\sigma^{\bar{\delta}(\Gamma)}) +\HH_{Q}(\sigma_{Q}) \right|}_{\textrm{I}} \\
+\underbrace{\left|-\HH_{Q}(\sigma_Q) -\frac 12 \EE_{Q}(g^N_{Q})\right|}_{\textrm{II}} +  \underbrace{\left|\frac 12 \EE_{Q}(g^N_{Q}) -\textrm{E}(Q)\right|}_{\textrm{III}}
\end{multline}
with $\sigma_{Q,y}$ was defined at \cref{E:GSS}.  Examining the term $Error_{Q}^N$ in \cref{L:BLayer} and recalling the conditions which entail $\Xi_{\ffrac \ell2}=1$
\[
\textrm{II}+ \textrm{III} \ls \epsilon^{\frac 94}.
\]
with plenty of room to spare for $\epsilon$ small.
Regarding Term \textrm{I}, using that $|\tau \theta|\leq 1$ and $\Xi_{\ffrac \ell2}(Q)=1$, if we only display dominant terms involving $\tau$, we have
\[
\textrm{I} \ls \|g^N_{Q}\|^2_{\infty}\|\nabla \tau\|_{2, Q}^2 + \|\nabla g^N_{Q}\|_{\infty} \|\nabla g^N_{Q}\|_{2} \|1-\tau\|_{2, Q}+ \epsilon^{\frac 94}.
\]
Let $W=\{x\in Q: \tau_x \neq 1\}$.  
Clearly 
\[
 \|1-\tau\|^2_{2, Q_{\ffrac \ell2}}\leq |W| \ls  \ell^{-\ffrac 12} \ell^3.
\]
and
\[
\|\nabla \tau\|_\infty \ls \ell^{-\ffrac 12}.
\]
Using the bounds $\|g^N_{Q}\|_{\infty}<2 \epsilon \ell^{\frac 12} |\log \epsilon|^{30}$ and $\|\nabla g^N_{Q_{\ffrac \ell2}}\|_{\infty} \leq \epsilon |\log \epsilon|^{30}$, 
\[
\textrm{I}\ls \epsilon^{\frac 94} \ell^3.
\]
Thus we have
\[
\sum_{Q \in \RR_{\ffrac \ell2}(\Gamma)}  \left[\textrm{E}(Q) + \HH_{Q}(\sigma^{\bar{\delta}(\Gamma)})\right] \mathbf 1\{ \Xi_{\ffrac \ell2}(Q)=1\} \ls  \epsilon^{\frac 94} |\Gamma|.
\]

Also,
\begin{align*}
\sum_{Q \in \RR_{\ffrac \ell2}(\Gamma)}  &\left[\textrm{E}(Q) + \HH_{Q}(\sigma^{\bar{\delta}(\Gamma)}) \right]\mathbf 1\{ \Xi_{\ffrac \ell2}(Q)=0\} \\ &\leq \underbrace{\sum_{Q \in \RR_{\ffrac \ell2}(\Gamma)}  \textrm{E}(Q)  \mathbf 1\{ \Xi_{\ffrac \ell2}(Q)=0,F^{\nabla}_{0}(Q) \leq \epsilon^2 |\log \epsilon|\}}_{\textrm{I}}\\
&+
 \underbrace{\sum_{Q \in \RR_{\ffrac \ell2}(\Gamma)}  \textrm E(Q) \mathbf 1\{ \Xi_{\ffrac \ell2}(Q)=0, F^{\nabla}_0(Q) \geq \epsilon^2 |\log \epsilon|\}}_{\textrm{II}}
\end{align*}
where we used that $\HH_{Q}(\sigma^{\bar{\delta}(\Gamma)})=0$ on these boxes.
Using \cref{H:FML,,E:rr5} and the fact that the density of boxes with $\Xi_{\ffrac \ell 2}=0$ is small to bound term \textrm{I} and using \cref{H:FML,,E:rr5,,E:rr1} to bound term \textrm{II}
\[
\textrm{I} + \textrm{II} \ls \epsilon^2 |\log \epsilon|^{-40} |\Gamma|.
\]
\end{proof}

\subsection{Proofs for \cref{S:Collar}}
\label{SS:7}
In this section we prove \cref{L:INT2,,L:INT3,,L:INT4}.  For notational convenience, we will use the functions $F_\lambda, F^{\nabla}_{\lambda}, R$ as defined after \cref{E:Dense} in our estimates with $\lambda = L^{-2} \log^8 L$.

\begin{proof}[Proof of \cref{L:INT2}]
We restrict the discussion to $\mathscr D^+$.  Using \cref{E:BL} along with the fact that $\sigma^1= \sigma^{2}$ on $\pO \mathscr D^+$ to cancel boundary contributions, we have
\begin{equation}
|-\HH_{\Lambda_N}(\sigma^{2}|e_1)+ \HH_{\Lambda_N}(\sigma^1|e_1)| \ls \EE_{\mathscr D'}(\sigma) + \EE_{\mathscr D'}(\gld_{\mathscr D^+}) + \lambda\|\gld_{\mathscr D^+}\|_{2}\sqrt{|\mathscr D^+|}
\end{equation}
where $\mathscr D' = \mathscr D^+ \cup \pO \mathscr D^+$.
By definition, the original spin configuration $\sigma$ has $\psi^{(0)}=1$ on $\mathfrak C(\Gamma)$ so 
\[
 \EE_{\mathscr D'}(\sigma)  \ls \epsilon^2 |\log \epsilon||\mathscr D^+|.
 \]
However if if $\Gamma$ is clean, 
\[
\frac{|\mathscr D^+|}{|\Gamma|} \ls |\log \epsilon|^{-55} 
\]
by condition \cref{E:Dense}.  Thus $\EE_{\mathscr D'}(\sigma) \ls \epsilon^2 |\log \epsilon|^{-50}$

Next, we have by the Cauchy-Schwarz inequality
\[
\sum_{x\in\DD} \left\{\sum_{\stackrel{Q\cap \mathscr D^+ \neq \varnothing}{Q \in \QQ_L}} R(Q)\left[e^{-c |\log \epsilon|^{4}}\wedge  e^{-c \lambda^{\frac 12} \dist(x, Q)}\right]\right\}^2 \\
\ls e^{-c |\log \epsilon|^{4}} \sum_{\stackrel{Q\cap \mathscr D^+ \neq \varnothing}{Q \in \QQ_L}}  R^2(Q) L^d
\]
and by \cref{E:rr4}
\[
 \sum_{\stackrel{Q \cap \delta (\Gamma)\neq \varnothing}{Q\in \QQ_L}}  R^2(Q) L^d \ls |\log \epsilon|^{100} |\Gamma|.
\]
Therefore, by \cref{L:TechBadRandomeness}, we have, for $\epsilon$ small enough,
\begin{align}
\label{E:SHT1}
 &\EE_{\mathscr D'}(\gld_{\mathscr D^+})  \ls \underbrace{L^d \sum_{Q \subset \mathscr D^+, \: Q\in \QQ_L} F_\lambda^{\nabla}}_{\textrm{I}} +  e^{-c|\log \epsilon|^{4}}|\Gamma|,\\
 &\|\gld_{\mathscr D^+}\|_2^2 \ls L^d \sum_{\stackrel{Q\cap \mathscr D^+ \neq \varnothing}{Q \in \QQ_L}} F_\lambda +  e^{-c|\log \epsilon|^{4}}|\Gamma|.
\end{align}

We bound the first term on the RHS of \cref{E:SHT1} via
\begin{align*}
\textrm{I}&\ls \underbrace{\epsilon^2|\log \epsilon| |\mathscr D^+|}_{\textrm{A}} + L^d \underbrace{\sum_{\stackrel{Q\cap \mathscr D^+ \neq \varnothing}{Q \in \QQ_L}} F_{\lambda}^{\nabla}(Q) 1\{F_{\lambda}^{\nabla}(Q) > \epsilon^2 |\log \epsilon|\}}_{\textrm{B}},\\
&\ls \epsilon^2 |\log \epsilon|^{-40} |\Gamma|.\end{align*}
To obtain the second bound, we used the fact that $|\mathscr D^+|\ls |D| \ls |\log \epsilon|^{-55}|\Gamma|$ to bound \textrm{A} and \cref{E:rr1} to bound \textrm{B}.

A similar estimate on $\|\gld_{\mathscr D^+}\|_2$ using \cref{E:rr2} in place of \cref{E:rr1} gives
\[
\lambda\|\gld_{\mathscr D^+}\|_{2, \mathscr D^+} \sqrt{|\mathscr D^+|} \ls \epsilon^{2} |\log \epsilon|^{-25} |\Gamma| + e^{-c|\log \epsilon|^4}|\Gamma| 
\]
Combining the estimates yields the claim in the Lemma.
\end{proof}

\begin{proof}[Proof of \cref{L:INT3}]
Note first of all that by construction,
\[
\EE_{\mathscr C^+}(\sigma^3) \ls \EE_{\mathscr C^+}(\sigma)+ \EE_{\mathscr C^+}(\gld_{\mathscr C^+})
\]
where $\sigma$ is the original spin configuration from which $\sigma^3$ is produced.
Therefore, \cref{L:BLayer} implies
\begin{equation*}
-\HH_{\Lambda_N}(\sigma^{3}|e_1)\geq -\HH_{\Lambda_N}(\sigma|e_1)
- C_1 Error^{\lambda, D}_{\mathscr{C}^+}(\sigma) + Error^{\lambda, D}_{\mathscr{C}^+}(\gld_{\mathscr{C}^+}) - C_2\epsilon^2 |\log \epsilon|^{-25} |\Gamma|
\end{equation*}

Because $\Xi_L (Q)=1$ for $\{ Q \in \QQ_L': Q\cap \mathscr C^+ \neq \varnothing\}$, inequalities \cref{E:r4}  and \cref{E:TechBadRandomeness1} allow us to exchange  terms involving $\gld_{\mathscr C^+}$ for sums involving $(\gld_Q)_{\stackrel{Q \cap \mathscr C^+ \neq \varnothing}{Q\in \QQ_L'}}$ up to errors of order $e^{-c \log^4 \epsilon}$.  Thus  
\begin{equation*}
Error^{\lambda, D}_{\mathscr{C}^+}(\sigma) + Error^{\lambda, D}_{\mathscr{C}^+}(\gld_{\mathscr{C}^+}) \ls \sum_{\stackrel{Q \cap \mathscr C^+ \neq \varnothing}{Q\in \QQ_L'}} Error^{\lambda, D}_{Q}(\sigma) + Error^{\lambda, D}_{Q}(\gld_{Q}) + \epsilon^2 |\log \epsilon|^{-25} |\Gamma|
\end{equation*}

Recall that $\EE_Q(\sigma) \leq \epsilon^2|\log \epsilon| |Q|$ if $Q \cap \mathscr C^+ \neq \varnothing$ because $Q$ is a cube at the outer boundary of a contour.  Then because  $\Xi_L (Q)=1$, inequalities \cref{E:r1}, \cref{E:r2}, and \cref{e:r3} imply
\begin{equation}
\label{E:ErrD3}
Error^{\lambda, D}_{Q}(\sigma) + Error^{\lambda, D}_{Q}(\gld_{Q}) \ls \epsilon^{\ffrac {9} 4}|Q|
\end{equation}
for $\epsilon$ small enough and the claimed lower bound follows.
\end{proof}

\begin{proof}[Proof of \cref{L:INT4}]
To begin with
\begin{equation}
\label{E:BCollar1}
\left|\HH_{\Lambda_N}(\sigma^{3}|e_1)-\HH_{\Lambda_N}(\sigma^{\mathfrak C}|e_1)\right| 
\ls \underbrace{\left|\EE_{W}(\sigma^{3})-\EE_{W}(\sigma^{\mathfrak C})\right|}_{\textrm{I}} + \underbrace{\left| \sum_{x\in \mathfrak C}\epsilon \alpha_x[\sigma^{3}_x-\sigma^{\mathfrak C}_x] \cdot e_2 \right|}_{\textrm{II}}.
\end{equation}
To estimate these terms we need some bounds on the interpolation.  First, let
\begin{equation}
\label{E:DefW}
W =\{ x: \exists y\sim x, \sigma^{\mathfrak C}_y\neq  \sigma^{3}_y \}
\end{equation}
denote the support of the modification.  Clearly
 \begin{align}
 \label{E:inerp1}
& |W|  \ls  \ell^{-\ffrac 12}  |\Gamma| \text{ and }\\
\label{E:inerp2}&|\nabla_{\langle xy \rangle} \tau| \ls \ell^{-\ffrac 12} 
\end{align}

For Term \textrm{I}, we have, using \cref{C:Contract},
\begin{equation}
\textrm{I}  \ls \EE_{W}(\sigma^3) +  \sum_{\langle x y \rangle \subset W} {\gld_{\mathscr C, x}}^2 [\nabla_{\langle xy \rangle} \tau]^2+ e^{-c|\log \epsilon|^4}|\Gamma|.
\end{equation}
For Term \textrm{II}, write $\epsilon\alpha= [-\Delta+ \lambda]\cdot g^{\lambda, D}_{\mathscr C}$on $\mathscr C$.
Summation-by-parts gives (in particular using that  since $W \subset \mathscr C$ is strict to cancel boundary terms)
\[
\textrm{II}
\leq
\underbrace{\left|\sum_{e \subset  W} \nabla_e g^{\lambda, D}_{\mathscr C} \nabla_e e_2\cdot[\sigma^{3}-\sigma^{\mathfrak C}])\right|}_{\textrm{III}} + 2 \lambda \sum_{x\in W} |g^{\lambda, D}_{\mathscr C, x}|
\]
We then have, using \cref{C:Contract},
\[
\textrm{III} \ls \EE_{W}( \sigma^{3})
+ \EE_{W}(g^{\lambda, D}_{\mathscr C}) + \sum_{\langle x y \rangle \subset W} {\gld_{\mathscr C, x}}^2 [\nabla_{\langle xy \rangle} \tau]^2 + e^{-c|\log \epsilon|^4}|\Gamma|.
\]

We now derive bounds on the respective RHS's.
By \cref{C:Contract} and the definitions of $W, \sigma^3$,
\begin{equation}
\label{E:D3Bound}
\EE_{W}(\sigma^{3}) \ls \EE_{W}(g^{\lambda, D}_{\mathscr C})  + e^{-c\log^4\epsilon}|\Gamma|.
\end{equation}
Since $\mathscr C \subset \cup_{\Xi_L(Q)=1} Q$,  \cref{L:TechBadRandomeness} implies
\begin{align}
\label{E:GApprox}
& \lambda \sum_{x\in W} |\gld_{\mathscr C, x}| \ls \epsilon^{\ffrac 94}|\Gamma|,\\
&\EE_{ W}( g^{\lambda, D}_{\mathscr C}) \ls \epsilon^{\ffrac 94}|\Gamma|.
\end{align}
because $\|\gld_{Q}\|_{\infty} \leq \epsilon\sqrt{L} \log^{30} \epsilon$ and $\|\nabla \gld_{Q}\|_{\infty} \leq \epsilon |\log \epsilon|^{30}$ if $\Xi_L(Q)=1$ and also using \cref{E:inerp1}.  Similarly, using \cref{E:inerp1,,E:inerp2}
\[
 \sum_{\langle x y \rangle \subset W}  {\gld_{\mathscr C, x}}^2 [\nabla_{\langle xy \rangle} \tau]^2 \ls \epsilon^{\ffrac 94} |\Gamma|
\]
and the claim is proved by collecting these estimates.
\end{proof}

\subsection{Proofs for \cref{S:Peierls}}
\label{S:PP}
Recall that we defined $\QQ_\ell^s=\{Q(r): Q(r-(\ffrac \ell 2,  \ffrac \ell 2,  \ffrac \ell 2) \in \QQ_{\ell}\}$.
\begin{proof}[Proof of \cref{L:EComp}]
Let $\sigma \in \X(\Gamma)$ and set \[
\begin{array}{lr}
\RR_{\ell}(\Gamma)=\{Q\in \QQ_\ell: Q \cap \bar{\delta}(\Gamma) \neq \varnothing\}, & 
\RR^s_{\ell}(\Gamma)=\{Q\in \QQ^s_\ell: Q \cap \bar{\delta}(\Gamma) \neq \varnothing\}.
\end{array}
\]
We first argue that enough of the energetic defect  of $\sigma$ is captured by restricting attention to its behavior to boxes either in $\RR_{\ell}(\Gamma)$ or in $\RR_{\ell}^s(\Gamma)$.

Let
\begin{align*}
&\mathcal S_{\EE}(\Gamma)= \{ Q \subset \Gamma \text{ such that $Q\in \QQ_L$ and $\Psi_{x_0}(Q)=0$ because $\psi^{(0)}_z=0$ for some $z$}\},\\
&\mathcal S_{\textrm{Ave}}(\Gamma)= \{ Q \subset \Gamma \text{ such that $Q\in \QQ_L$ and $\Psi_{x_0}(Q)=0$ but $\psi^{(0)}_z=1$ } \text{ for all relevant $z$}\}.
\end{align*}
If $Q \in \mathcal S_{\EE}(\Gamma)$ there exists a cube of side-length $\ell$, $Q_{\ell}$, so that $\dist(Q_{\ell}, Q) \leq 2L + 5\ell$ and
\[
\EE_{Q_{\ell}}(\sigma) \geq  \epsilon^2 |\log \epsilon| \ell^3.
\]
Then there must be ${Q'}$ in either $\RR_{\ell}(\Gamma)$ or $\RR^s_{\ell}(\Gamma)$ so that
\begin{equation}
\label{E:Bad1}
\EE_{Q'}(\sigma) \geq \ffrac {\epsilon^2}{16} |\log \epsilon| \ell^3.
\end{equation}
Note that these considerations include cubes which overlap $\Lambda_N^c$ in the case of $\RR^s_{\ell}(\Gamma)$.
Let
\begin{align}
&A_1=\{Q \in \RR_{\ell}(\Gamma) \text{ such that \cref{E:Bad1} holds}\}\\
&A_2=\{Q \in \RR^s_{\ell}(\Gamma) \text{ such that \cref{E:Bad1} holds}\}
\end{align}

On the other hand if $Q \in \mathcal S_{\text{Ave}}(\Gamma)$ there exits $Q_\ell$ so that $\dist(Q_\ell, Q_{L}(x_0)) \leq 2L + 5\ell$ and
\[
|\sigma(Q_\ell) \cdot e_1| \leq 1- \xi.
\]
Then there must be $Q'$ in $\RR_{\ell}(\Gamma)$ 
\begin{equation}
\label{E:Bad2}
|\sigma({Q}')\cdot e_1| \leq 1-\frac{\xi}{16}
\end{equation}
Let
\[
A_3=\{Q \in \RR_{\ell}(\Gamma) \text{ such that \cref{E:Bad2} holds}\}.
\]
Because $\sigma \in \mathbb X(\Gamma)$, 
\[
\max_{i} |A_i| \gtrsim |\log \epsilon|^{-24}N_L^{\Sp(\Gamma)}.
\]

Now we compare the internal energy of $\sigma$ with that of $\textrm{S}^{\pm}$. 
Recall that $ext$ denotes the boundary condition which is set to $e_1$ on $\pO R \cap \Lambda_N^c$ and is free otherwise.  Then reflection invariance of the Hamiltonian for components with free boundary conditions implies 
\[
-\HH_{\Lambda_N \backslash \bar{\delta}(\Gamma)}(\sigma^{\mathfrak C}|\textrm{ext})= -\HH_{\Lambda_N \backslash \bar{\delta}(\Gamma)^*}(\sigma^*|\textrm{ext}).
\]
Thus, by \cref{L:Collar}
\begin{multline}
-\HH_{\Lambda_N}(S^{\pm}_\Gamma| e_1) + \HH_{\Lambda_N}(\sigma| e_1) \geq  \underbrace{-\HH_{ \bar{\delta}(\Gamma)}(\sigma^{\bar{\delta}(\Gamma)}|\textrm{ext})  + \HH_{\bar{\delta}(\Gamma)}(\sigma^{\mathfrak C}|\textrm{ext})}_{\textrm{I}} - C\epsilon^2|\log \epsilon|^{-25}|\Gamma|.
\end{multline}

Next recall that
\[
\textrm{E}_0(Q)=\max_{\sigma} -\HH_Q(\sigma|\textrm{ext}).
\]
By \cref{L:GSd3},
\[-\HH_{ \bar{\delta}(\Gamma)}(\sigma^{\bar{\delta}(\Gamma)}|\textrm{ext})\geq \sum_{Q\in \RR_{\ffrac \ell 2}(\Gamma)} \textrm{E}_0(Q)- \epsilon^2|\log \epsilon|^{-40}|\Gamma|
\]
Note that if we have a finite region $R= R_1\cup R_2$ with $R_1 \cap R_2 = \varnothing$ then
\[
\textrm{E}_0(R)\leq \textrm{E}_0(R_1)+\textrm{E}_0(R_2).
\]
Note also that if $Q\in A_i$, $Q\cap \Lambda_N$ may be covered by boxes in $\RR_{\ffrac \ell2}(\Gamma)$.
Therefore, for $i \in \{1, 2, 3\}$,
\[
\textrm{I} \geq \underbrace{\sum_{Q\in A_i} \left[\textrm{E}_0(Q\cap \Lambda_N)+ \HH_{Q\cap \Lambda_N}(\sigma|\textrm{ext})\right]}_{\textrm{II}_i}- \epsilon^2|\log \epsilon|^{-25}|\Gamma|
\]
The superscript $\mathfrak C$ was dropped because the original spin configuration $\sigma$ is unmodified on such boxes as they must intersect $\textrm{sp}(\Gamma)$.  
It remains to show that 
\begin{equation}
\label{E:maxim}
\max_{i}\textrm{II}_i \gtrsim \epsilon^2|\log \epsilon|^{-24} |\Gamma|.
\end{equation}

We consider only $A_2$ in detail.  The remaining cases are a bit simpler to handle as we do not have to deal with cubes overlapping $\Lambda_N^c$.

Since we can always choose $\sigma_y \equiv e_1$, $\textrm{E}_0(Q\cap \Lambda_N)\geq 0$.  If $Q\in A_2$ and $\Xi_{\ell}(Q)=1$, \cref{H:FML} of \cref{L:FBE} implies
\[
- \HH_{Q\cap \Lambda_N}(\sigma|\textrm{ext})\ls -\epsilon^2|\log \epsilon| \ell^3
\]
because we may view $- \HH_{Q\cap \Lambda_N}(\sigma|\textrm{ext})$ as $- \HH_{Q}(\sigma')$ for the configuration 
\[
\sigma'_y=
\begin{cases}
\sigma_y \text{ if $y \in Q\cap \Lambda_N$}\\
e_1 \text{ if $y \in Q\cap \Lambda_N^c$}.
\end{cases}
\]
Thus
\[
\sum_{Q\in A_2: \Xi_{\ell}(Q)=1} \left[\textrm{E}_0(Q)+ \HH_{Q}(\sigma|\textrm{ext})\right] \gtrsim 
\epsilon^2|\log \epsilon| \ell^3 |A_2|.
\]

On the other hand, using \cref{H:FML} once again,
\[
\sum_{Q\in A_2: \Xi_{\ell}(Q)=0} \left[\textrm{E}_0(Q)+ \HH_{Q}(\sigma|\textrm{ext})\right] \gtrsim 
- \underbrace{\sum_{Q\in A_2: \Xi_{\ell}(Q)=0} \left[ \EE_Q(g^N_Q)+ \epsilon|\alpha(Q)|\ell^3\right]}_{\textrm{III}}
\]
Using conditions \cref{E:rr1,E:rr5} and that fact that $\delta(\Gamma)$ is good at the scale $\ell$,
\[
\textrm{III} \ls \epsilon^2 |\log \epsilon|^{-50}|\Gamma|.
\]
It follows that
\[
\textrm{II}_2 \gtrsim \epsilon^2|\log \epsilon| \ell^3 |A_2|- \epsilon^2 |\log \epsilon|^{-50}|\Gamma|.
\]
The same sort of argument gives
\begin{align}
&\textrm{II}_1 \gtrsim \epsilon^2|\log \epsilon| \ell^3 |A_1|- \epsilon^2 |\log \epsilon|^{-50}||\Gamma| \\
&\textrm{II}_3 \gtrsim \xi^2 \epsilon^2 \ell^3 |A_3|- \epsilon^2 |\log \epsilon|^{-50}||\Gamma|
\end{align}
where we have used \cref{E:FL} instead of \cref{H:FML} to estimate the contribution from $A_3$.
\Cref{E:maxim} and hence \cref{L:EComp} follow.
\end{proof}

\section{Estimates on the Randomness}
\label{S:Randomness}
In this section we derive elementary probabilistic estimates which underly the rest of the paper.
Let $l \in \N$ and restrict $\alpha$ to $Q_{l}\subset \Z^d$.  With an eye toward future work we record bounds for all $d \geq 2$ though we only use the $d=3$ case here.
\begin{align*}
&\Gld_{Q_l}=\frac{\gld_{Q_l}}{\epsilon}= (-\Delta^{D}_{Q_l}+ \lambda)^{-1} \cdot\alpha,\\
&\Gln_{Q_l}=\frac{\gln_{Q_l}}{\epsilon}=(-\Delta^{N}_{Q_l}+ \lambda)^{-1} \cdot [\alpha- \alpha_{Q_{l}}].
\end{align*}
Here $\lambda \in [0, 1)$.
From now on $Q_l$ is fixed and we drop this subscript from our notation.

It is important for us to have fairly precise probabilistic estimates on the quantities
\[
\|\nabla \Gld\|_{\infty}, \:\:
\|\Gld\|_{\infty}, \:\:
\|\Gld\|_2,\:\:
\|\nabla \Gld\|_2,
\]
and similarly for $\Gln $, to have estimates on $ \alpha_{Q_l}\sqrt{|Q_l|}$ and probabilistic bounds on fluctuations of the low momentum modes.
In $d=2$ we record somewhat more refined information, in particular bounds on the density of points with atypical fluctuations.

For general finite regions $R \subset \Z^d$, we introduced potentials of the form
\[
m_x = \sum_{\stackrel{e\cap R \neq \varnothing}{x\in e}} [\nabla_e \gld_R]^2
\]
in \cref{S:Energetics} and used its typical behavior heavily in \cref{S:Collar}.  Here we specialize $R= Q_l$.  For notational convenience set
\[
R_\lambda(x)=\lambda^{-\frac 12} \wedge \dist(x, \pO Q_l)
\]
Given $A>0$, consider the event
\[
\mathcal A_r=\mathcal A_r(A, Q_l)=\left\{\omega: r^{-d}\sum_{\|y-x\|_\infty \leq r} m_y \geq A \epsilon^2\log^{\delta_{2, d}} R_{\lambda}(x)  \text{ for all $x\in Q_l$ s.t. }{\dist_{\infty}(x, \pO Q_l) \geq \frac{l}{16} } \right\}
\]
where $\delta_{2, d}$ is $1$ if $d=2$ and $0$ otherwise.

Let
\begin{align*}
&\varsigma_2^2:=\varsigma_{2, \lambda, l}^2=\int_{[l^{-1}, 2 \pi]^d} \textrm{d}^dk (\|k\|^2+ \lambda)^{-2},\\
&\varsigma_{\nabla}^2:=\varsigma_{\nabla, \lambda, l}^2=\int_{[l^{-1}, 2 \pi]^d}\frac{\|k\|^2}{(\|k\|^2+ \lambda)^{2}},
\end{align*}

The following lemma summarizes the bounds we need.  
\begin{lemma}
\label{L:RandBasic}
Let $d \geq 2$ and $\delta \in (0, 1)$ be fixed.  For any $\lambda \in [0, 1)$ we have the following probabilistic estimates for either choice of boundary conditions:
\begin{enumerate}
\item
Let $M \in (1, \infty)$ be fixed.  For any $x \in Q_{l}$ and any edge $e$
\begin{align*}
&\bP\left(|\Gl_x| \geq  M \varsigma_2\right) \ls e^{-cM^2}\\
&\bP\left(|\nabla_e \Gl| \geq M \varsigma_{\nabla}\right) \ls e^{-cM^2}.
\end{align*}

\item 
\[
\mathbb P\left(\|\Gl\|_2^2 \geq m \varsigma_{2}^2 l^d\right) \ls \begin{cases}e^{-cm l^{d-4-\delta}} \quad \text{ for $d \geq 5$},\\
l^{-c m}\quad \text{ for $d=4$},\\
 \log \log l e^{-cm} \text{ for $ d\in \{2,3\}$}.
\end{cases}
\]

\item
For all $d \geq 2$, there is $c_d>0$ so that
\[
\mathbb P\left(\|\nabla \Gl\|_2^2 \leq c_d \varsigma_{\nabla}^2 l^d\right) \ls 
e^{-c l^{\ffrac d2}}.
\] 
\item
\[
\mathbb P\left(\|\nabla \Gl\|_2^2 \geq m \varsigma_{\nabla}^2 l^d\right) \ls \begin{cases}
e^{-cm l^{d-2-\delta}}, \quad \text{for $d \geq 3$},\\
\log \log l e^{-c m} \quad \text{ for $d=2$}.
\end{cases}
\]

\item
\label{L:AvePot}
There is $A_d>0$ so that if $0\leq \lambda \leq l^{-1}$, 
\[
\mathbb P(\cup_{\{\log^{90} l \leq r \leq \ffrac l4 \}}\AA_r^c(A_d)) \ls e^{-c\log^{60} l}.
\]
\end{enumerate}
\end{lemma}
We will say a bit in the way of proof about these bounds in a bit, but let us first address \cref{P:Dirt4} and \cref{L:BadSetBound}.

\begin{proof}[Proof of \cref{P:Dirt4}]
We show 
\[
\mathbb P\left( \left[F_\lambda^{\nabla};  \mathbf 1_{\{F_\lambda^{\nabla} \geq  \epsilon^2 |\log \epsilon|\}}\right]_{Y} \geq 8^6 \epsilon^{\ffrac 94} N_Y^{L_0} \right) \ls e^{-c\log^2 \epsilon N_{L_0}^Y}
\]
only.
All remaining bounds are proved in a similar or simpler way using \cref{L:RandBasic}.
For each $\eta \in \frac{{L_0}}{16} \{- 32, \dotsc, 32\}^3$ let
\[
\FF_\eta=\left\{Q\left(r+ \eta\right): Q \subset Y, Q\in \QQ_{L_0}\right\}
\]
for some choice of $\eta \in \frac{{L_0}}{16} \{- 32, \dotsc, 32\}^d$. Let 
\[
m(Q)= \min\{ m\in \N: \: \|\nabla \gll_{Q}\|^2_{2}  <  m \epsilon^2 |\log \epsilon|{L_0}^3\}
\]
To obtain the estimate, it is enough to bound probabilities for the finite collection of events indexed by $\eta$
\[
\left\{\sum_{Q \in \FF_\eta}m(Q) \mathbf 1\{ m(Q) \geq  2\} \\
\geq \epsilon^{\ffrac 94} N_Y^{L_0}\right\}
\]

For any subset $\AA\subset \FF_{\eta}$, any choice $\{n(Q): Q \in \AA, n(Q) \geq 2\}$ and for $\epsilon$ small enough, 
\[
\mathbb P(m_{Q}= n(Q) \: \forall Q \in \AA) \leq C^{|\AA|}e^{-c\sum_{Q_{\ell} \in \AA}n_{Q_\ell}|\log \epsilon| {L_0} }
\]
This implies, after some standard computations to take into account the entropy of the family of subcubes which contribute to the sum, that
\begin{equation}
\label{E:StupidBound}
\mathbb P\left(\sum_{Q \in \FF_\eta} m(Q) \mathbf 1\{m(Q) \geq 2\} \geq  M N^{L_0}_Y\right) \leq C^{N^{L_0}_Y} e^{-cM |\log \epsilon| {L_0} N^{L_0}_Y}
\end{equation}
for any $M > 0$. Taking $M= \epsilon^{\ffrac 34}$ the exponent on the RHS is still at least $ \epsilon^{-\ffrac 14}N^L_Y$ and the claim follows.

\end{proof}

\begin{proof}[Proof of \cref{L:BadSetBound}]
We claim first of all that
\[
\mathbb P\left(A(x)\right) \leq
Ce^{-c|\log \epsilon|^{2}}
\]
and
\begin{equation}
\label{E:Cor}
\left|\mathbb P\left(A(x)\cap A(y)\right)-\mathbb P\left(A(x)\right) \mathbb P\left( A(y)\right)\right|  \leq
C e^{-c\log^2 \epsilon\dist_{L}(x,
y)}
\end{equation}
Here $\dist_{L}(x, y)$ denotes the minimal number of
blocks in an $L$-measurable block path from $Q(x)$ to $Q(y)$.

Consider the collection of
bounded $L$-measurable connected subsets of $\Z^3$ containing
$x$, $\{Y\}_{ Y \ni x}.$
For each such $Y$, we can apply the probabilistic estimates of \cref{P:Dirt4}. The number of $L$-measurable connected sets $Y$
containing $x$ with $N^L_Y=r$ is well known to have the asymptotic
$a_0^r$ for some fixed, dimension dependent constant $a_0$. Thus
\[
\mathbb P(Q(x) \text{ is in some $Y$ which is not clean})\leq
C \sum_{r \geq 1} (2a_0)^{r} e^{-c\log^2 \epsilon r} \quad \text{ if $d=3$}
\]
Modifying
this estimate slightly via the discrete isoperimetric inequality,
\begin{equation} \label{E:ClustBound} \mathbb P(A(x) \text{ is in
$c(Y)$ for some $Y$ which is dirty})\leq
C\sum_{r \geq 1}
r^{d/(d-1)}(2a_0)^{r} e^{-c\log^2 \epsilon r} \quad \text{ if $d=3$}.
\end{equation}
The first claim follows.

Next we prove the correlation bound.  Using the fact that
the events  $\{Y_i \text{ is } \text{dirty}\}$ are
independent if $\delta_{2L}(c(Y_1)) \cap \delta_{2L}(c(Y_2)) = \varnothing$ we have
\begin{multline}
\label{E:Cor1}
\left|\mathbb P(A(x_1), A(x_2))-\mathbb P(A(x_1)) \mathbb P(A(x_2))\right|
 \leq \\
 \mathbb P(Q(x_1), Q(x_2) \text{ are in $c(Y_1), c(Y_2)$ for some dirty $Y_1, Y_2$ with $\delta_{2L}(c(Y_1)) \cap \delta_{2L}(c(Y_2)) = \varnothing$}).
 \end{multline}
An estimate similar to \cref{E:ClustBound} then gives
\begin{equation}
\label{E:Cor2}
\left|\mathbb P(A(x_1), A(x_2))-\mathbb P(A(x_1)) \mathbb P(A(x_2))\right| \leq
C|Q_L| e^{-c\log^2 \epsilon \dist_L(x_1, x_2)}.
\end{equation}

The first bound implies
\[
\mathbb E[ |\mathbb D_N|] \ls |\Lambda_N|e^{-c|\log \epsilon|^{2}}
\]
The correlation bound implies that
\[
\text{Var}\left[|\D_N|\right] \ls L^3|\Lambda_N| e^{-c|\log \epsilon|^{2}} 
\]
By taking $N=2^{k}$, applying Chebyshev's inequality to estimate deviations of
$|\mathbb D_N|$ and then the Borel-Cantelli lemma along this
subsequence, we have that, for almost every $\omega\in \Omega$, there
is $N_0(\omega)\in \N$ so that
\[
 \frac{|\mathbb D_N|}{|\Lambda_N|}\leq
C e^{-c|\log \epsilon|^{2}}
\]

\end{proof}

\begin{proof}[Proof of \cref{L:RandBasic} (1),(2), (3), and (4)]\
These bounds are elementary computations and we only sketch the basic argument.  Statement (1) simply relies on the fact that $G^{\lambda}_x, \nabla_e G^{\lambda}$ are Gaussian variables with mean zero and variance bounded above by a constant multiple of $\varsigma_2, \varsigma_{\nabla}$ respectively.  The idea of (2), (3) and (4) is to expand the field $G^{\lambda}$ in terms of either Dirichlet or Neumann Laplacian eigenfunctions depending on the boundary conditions.  By (a relative of) Parseval's identity, we may pass to momentum space and express all quantities  of interest as weighted sums of squares of i.i.d standard Gaussian variables (i.e. weighted by the eigenvalues of $[-\Delta + \lambda]^{-2}, -\Delta [-\Delta + \lambda]^{-2}$).  Here is where, for convenience, we use the Gaussian assumption: Fourier transforms of i.i.d. Gaussians are i.i.d. Gaussians.

These latter sums are estimated by first separating summands according to the momentum space annuli 
\[
A_s:=\{k\in \ffrac {2\pi} l \{1, \dotsc l\}^d: 2^{-(s+1)} \leq \|k\|_2 \leq  2^{-{s}}\}
\]  
This is useful because eigenvalues of $\Delta$ corresponding to these momenta are the same up to a multiplicative constant independent of $s$ and we can treat the contribution from each annulus as a constant multiple of an i.i.d. sum of squares of Gaussians indexed by wave vectors in the annuli.  If the cardinality of $A_s$ is big enough, the corresponding sum is highly concentrated around its mean while if the cardinality of $A_s$ is small the corresponding sum of squares must have a reasonably large fluctuation to contribute to the overall summation.  The extent to which these sets really contribute is reflected in the various cases stated in the Lemma.
\end{proof}

\begin{proof}[Proof of \cref{L:RandBasic} (5)]
Let us define 
\begin{align}
&\HH_r= \{ Q \subset Q_l: \text{Q is a cube of sidelength $r$}\},\\
&F=\cap_{r \geq \log^{90} l} \cap _{Q \in \HH_r} \{ \omega: \|\nabla \gld_Q\|_2^2 \geq c_d \epsilon^2 \varsigma_{\nabla, \lambda, r}^2 r^d\},\\
&F_1=\{ \omega: \|\nabla_e \gld_{Q_l}\|_{\infty, Q_l\cup \pO Q_l} \leq \epsilon \log^{30} l \}.
\end{align}
Then by \cref{L:RandBasic} (1),(3), 
\[
\mathbb P(F^c \cup F_1^c) \ls \exp(-c \log ^{55} l).
\]
We prove that $F\cap F_1 \subset \AA_r(A_d)$ for appropriate choices of $r, A_d>0$. On $F\cap F_1$, 
let $Q\in  \cup_{r \geq \log^{90} l}   \HH_r$.  For $x \in Q$, we may express the field $\gld_{Q_l, x}$ via
\[
\gld_{Q_l, x}= \gld_{Q, x}+ g^{(1)}_x  
\]
where $g^{(1)}$ is satisfies the Laplace equation $-\Delta g^{(1)}\equiv 0$ on $Q$ subject to the boundary condition $g^{(1)}_x= \gld_{Q_l, x}$ for  $x \in \pO Q$.  

Notice that
\[
\sum_{e \cap Q \neq \varnothing} [\nabla_e(\gld_{Q, x}+ g^{(1)}_x  )]^2 = \EE_{Q\cup \pO Q}(\gld_{Q, x})+ 
\sum_{e \cap Q \neq \varnothing} [\nabla_e g^{(1)}_x  ]^2
\]
since the cross term vanishes.  This is because $g^{(1)}$ is harmonic in $Q$ and $\gld_{Q}$ vanishes on $\pO Q$.
Hence
\[
\sum_{x \in Q} m^2_x \geq \EE_{Q\cup \pO Q}(\gld_Q)\geq c_d \epsilon^2 \varsigma_{\nabla, \lambda, r}^2 r^d
\]
because we restricted attention to $F\cap F_1$.  Thus $F\cap F_1\subset \AA_r(c_d, Q_l)$ whenever $\log^{90} l \leq r \leq \ffrac l4$.
\end{proof}

\noindent \textbf{Acknowledgements:} I'd like to thank
M.~Biskup, D.~Ioffe, G.~Kozma and S.~Shlosman for helpful discussions on this and related problems during various stages of its development.

\vfill

\noindent {\sc Nicholas Crawford:}
{\tt nickc@tx.technion.ac.il}\\

\end{document}